\documentclass[12pt]{article}
\usepackage[left=2.5cm,right=2.5cm,
top=3cm,bottom=3cm]{geometry}
\usepackage[T2A]{fontenc}
\usepackage[utf8]{inputenc}
\usepackage[english]{babel}
\usepackage{amsmath,amssymb,amsthm,mathtools,bm}
\usepackage{euscript,mathrsfs}
\usepackage{icomma}
\usepackage{graphicx}
\usepackage{wrapfig}
\usepackage[dvipsnames]{xcolor}
\usepackage{indentfirst}
\usepackage{amsbsy}
\usepackage{wasysym}
\usepackage{hyperref}
\usepackage{cancel}
\usepackage{verbatim}
\usepackage{empheq}
\usepackage{adjustbox}
\usepackage{authblk}

\usepackage{tikz-cd}
\usepackage{tikz}
\usetikzlibrary{positioning}
\usetikzlibrary{calc}

\definecolor{mylightred}{RGB}{211,79,73}
\definecolor{mydarkred}{RGB}{199,44,38}
\definecolor{mylightgreen}{RGB}{78,153,67}
\definecolor{mydarkgreen}{RGB}{43,129,33}
\definecolor{mylightpurple}{RGB}{150,107,178}
\definecolor{mydarkpurple}{RGB}{126,78,160}
\definecolor{mylightblue}{RGB}{49,101,205}
\definecolor{mydarkblue}{RGB}{20,92,205}

\tikzset{
  juliadot/.style args={#1,#2}{shape=circle,line width=0.03ex,minimum width=0.4ex,fill=#1,draw=#2}
}

\newcommand\julialetter[1]{{\strut\fontfamily{cmss}\bfseries\selectfont{#1}}}

\DeclareRobustCommand\julia{%
\begin{tikzpicture}[baseline=0mm, every node/.style={inner sep=0mm, outer sep=0mm}]
\node[anchor=base]        (j) at (0,0) {\julialetter{\j}};
\node[anchor=base, right=0ex of j] (u) {\julialetter{u}};
\node[anchor=base, right=0ex of u] (l) {\julialetter{l}};
\node[anchor=base, right=0ex of l] (i) {\julialetter{\i}};
\node[anchor=base, right=0ex of i] (a) {\julialetter{a}};
\path let \p1 = (j) in node[juliadot={mylightblue,mydarkblue}] (bluedot) at (\x1+0.02ex,1.4ex) {};
\path let \p1 = (i) in node[juliadot={mylightred,mydarkred}] (reddot) at (\x1,1.4ex) {};
\path let \p1 = (reddot) in node[juliadot={mylightpurple,mydarkpurple}] (purpledot) at (\x1+0.5ex,\y1) {};
\path let \p1 = (reddot) in node[juliadot={mylightgreen,mydarkgreen}] (greendot) at (\x1+0.25ex,\y1+0.42ex) {};
\end{tikzpicture}%
}

\definecolor{urlcolor}{HTML}{120099}
\definecolor{linkcolor}{HTML}{005F5F}
\hypersetup{pdfstartview=FitH,  linkcolor=linkcolor,urlcolor=urlcolor, colorlinks=true,citecolor=blue}

\theoremstyle{plain}
\newtheorem{theorem}{Theorem}
\newtheorem*{theorem*}{Theorem}
\newtheorem{proposition}{Proposition}
\newtheorem*{proposition*}{Proposition}

\renewcommand{\phi}{\varphi}
\renewcommand{\epsilon}{\varepsilon}
\DeclareMathOperator{\im}{\mathrm{Im}}

\DeclareMathOperator{\tr}{\mathrm{tr}}
\def\wt#1{\widetilde#1}

\numberwithin{equation}{section}
\begin{document}

\title{\textbf{Phase-locking in dynamical systems and quantum mechanics}}
\author[1,2,6]{\normalsize Artem Alexandrov\thanks{aleksandrov.aa@phystech.edu}}
\author[4,5,6]{\normalsize Alexey Glutsyuk\thanks{aglutsyu@ens-lyon.fr}}
\author[1,3]{\normalsize Alexander Gorsky}
\affil[1]{\small Institute for Information Transmission Problems, Moscow, 127994, Russia}
\affil[2]{\small Moscow Institute of Physics and Technology, Dolgoprudny 141700, Russia}
\affil[3]{\small Laboratory of Complex Networks, Center for Neurophysics and Neuromorphic Technologies, Moscow, Russia}
\affil[4]{\small Higher School of Modern Mathematics, Moscow Institute of Physics and Technology, Dolgoprudny 141700, Russia}
\affil[5]{\small CNRS UMR 5669 (UMPA, ENS de Lyon), Lyon, France}
\affil[6]{\small HSE University, Moscow, Russia}

\date{}

\maketitle

\begin{abstract}
    In this study, we discuss the Pr\"{u}fer transform that connects the dynamical system on the torus and the Hill equation, which is interpreted as either the equation of motion for the parametric oscillator or the Schr\"{o}dinger equation with periodic potential. The structure of phase-locking domains in the dynamical system on torus is mapped into the band-gap structure of the Hill equation. For the parametric oscillator, we provide the relation between the non-adiabatic Hannay angle and the Poincar\'{e} rotation number of the corresponding dynamical system. In terms of quantum mechanics, the integer rotation number is connected to the quantization number via the Milne quantization approach and exact WKB. Using recent results concerning the exact WKB approach in quantum mechanics, we discuss the possible non-perturbative effects in the dynamical systems on the torus and for parametric oscillator. The semiclassical WKB is interpreted in the framework of a slow-fast dynamical system. The link between the classification of the coadjoint Virasoro orbits and the Hill equation yields a classification of the phase-locking domains in the parameter space in terms of the classification of Virasoro orbits.  Our picture is supported by numerical simulations for the model of the Josephson junction and Mathieu equation.
\end{abstract}

\newpage
\tableofcontents

\section{Introduction}\label{sec:Introduction}

The classical phase-locking phenomenon is a well-known fundamental property of nonlinear dynamical systems \cite{pikovsky2001synchronization}. It is familiar for the overdamped Josephson junction (JJ) driven by periodic external current (which is the famous RSJ model \cite{mccumber1968effect}), systems of microparticles~\cite{mishra2025phase}, superconducting nanowires~\cite{dinsmore2008}, charge density waves~\cite{gruner1985charge}, skyrmions~\cite{reichhardt2015shapiro} and other nonlinear systems. Domains in the parameter space of the system where phase-locking takes place are sometimes called Arnold tongues. The simplest possible example of systems where phase-locking occurs is the nonlinear dynamical systems on a two-dimensional torus. The central object in the description of these systems is the Poincar\'{e} rotation number. Another common object is the Cantor staircase, which is often called the Devil's staircase, being nothing more than the plot of rotation number as a function of certain system parameters, and steps at the staircase correspond to the phase-locking domains.

The familiar example of the staircase structure is provided by the integer quantum Hall effect (QHE) when the Hall conductivity $\sigma_{xy}$ takes integer values on the plateau as a function of the external magnetic field. Similar plateaus at rational numbers occur for the fractional QHE as well. In this case, the quantization phenomenon on the staircase plateau can be interpreted as the quantization of the Berry curvature flux through the domain of the parameter space. For the QHE on the lattice, the quasi-momenta $(k_x,k_y)$ play the role of the parameters that define the Berry curvature flux through the Brillouin zone \cite{thouless1982quantized}. In the continuum case of QHE on the Riemann surface the corresponding parameters are identified as the fluxes through the noncontractable cycles or the moduli of the complex structures \cite{avron1985quantization,avron1995viscosity}. The staircase picture for several examples has been interpreted as the specific two-coupling renormalization group (RG) flow pattern when the non-perturbative renormalization of the couplings dominates the RG flow \cite{flack2023generalized}. In that case, the first-order RG equations are mapped to the second-order Schr\"{o}dinger equation with a peculiar potential with a double periodic structure.

The simplest possible example of the dynamical system where the phase-locking takes place is the phase oscillator forced by an external periodic force. The phase space of such system is the 2D torus and, under mild assumptions, instead of flow on torus, one can consider the standard circle map. This situation is typical because dynamical systems on 2D torus can be investigated by the analysis of circle diffeomorphisms. In the case of a phase oscillator with periodic force, phase-locking occurs if the rotation number is rational.

As noticed in the book \cite{arn}, a randomly sampled circle diffeomorphism is usually close to rotation, so it has an irrational rotation number. This implies that Arnold tongues are usually assumed to be quite small areas in parameter space. However, there is a class of dynamical systems on torus, where Arnold tongues are large, i.e. phase-locking takes place for quite large domains of parameter space. The ultimate feature of systems in this class is that phase-locking occurs \emph{only} for \emph{integer} rotation numbers. In the following, we will call such systems as of M\"{o}bius type. This phenomenon is commonly called \emph{rotation number quantization} \cite{buchstaber2010rotation}. It is amusing to note that for the RSJ model quantization of rotation number is indeed quantization because the rotation number in this system up to numerical factor (that contains Planck constant) coincides with the time-averaged voltage on JJ. The related Cantor staircase is known in the RSJ model as \emph{Shapiro steps}~\cite{shapiro1963josephson}. In~\cite{renne1974some,waldram1982alternative} the absence of fractional Shapiro steps has been explained. For the RSJ model the structure of phase-locking domains can be investigated by analysis of the Heun equation and its monodromy, which was done in papers \cite{buchstaber2017monodromy,glutsyuk2014adjacency,glutsyuk2019constrictions}. The construction of solutions for the RSJ model is discussed in \cite{tertychniy2006,buchstaber2013explicit,buchstaber2015holomorphic}. In addition, consideration of the RSJ model with two incommensurable periods has been made in \cite{bondeson1985}. Note that systems of M\"{o}bius type on 2D torus also appear in the investigation of bicycle dynamics \cite{bizyaev2017hess}, but without analysis of the phase-locking phenomenon.

In this study, we take advantage of the Pr\"{u}fer transform, which brings the second-order differential equation into the system of two first-order equations. This provides the correspondence between the Hill equation and the dynamical system on the 2D torus and vice versa. In other words, the potential in the corresponding Hill operator is defined uniquely by the flow on a two-dimensional torus. For the RSJ model, it has been discussed in \cite{renne1974some}, while the more general case was discussed in \cite{johnson1982rotation}. The main feature of this mapping is that it provides the connection between the Poincar\'{e} rotation number and the properties of the Hill equation solutions.

We explore the relations between dynamical systems on the torus of M\"{o}bius type, the equation of motion for the classical parametric oscillator and the band structure in quantum mechanics with periodic potential. We start our narration with a brief review of dynamical systems on the 2D torus and recall the essential properties of systems with rotation number quantization effect. Using the previously obtained results, we recall how the rotation number quantization is connected to the properties of Hill equation. We apply this connection to the classical parametric oscillator and find the relation between the Hannay angle of the classical Hamiltonian system, which is the classical counterpart of the quantum Berry phase, and the rotation number. The rotation number quantization will be linked to the quantization condition using the Milne approach. The semiclassical approximation in quantum mechanics will be linked with the slow-fast dynamics in the classical systems.

Since the Hill operator governs the classification of the coadjoint Virasoro orbits, we use this classification for the analysis of the phase-locking domains in the RSJ model. It turns out that the so-called special Virasoro orbits specified by the stabilizer vector fields with multiple zeros are relevant in this case. We identify the relation of the Hannay angle to the parameters of the coadjoint orbits. 

We show that slow-fast dynamics in the dynamical system on the torus can be mapped to the analysis of the Schr\"{o}dinger equation with the periodic potential.  The Schr\"{o}dinger equation suggests the application of recently developed methods to sum up non-perturbative contributions in dynamical systems. In these cases, the notion of the non-perturbative effect deserves explanation. In the parametric oscillator case, such non-perturbative effects can be attributed to the so-called complex-time instantons used a long time ago for the discussion of the particle in time-dependent fields \cite{brezin1970pair,popov1972pair}. On the other hand in the theory of dynamical systems we suggest the new interpretation of canards, the phenomenon when non-perturbative effects in the smallness parameter in the slow-fast system~\cite{diener1984canard,desroches2011canards} emerge. For the dynamical system on the torus, the presence of canards does not require additional parameters, which was shown and proven in~\cite{guckenheimer2001duck} and an investigation of their properties was carried out in~\cite{shchurov2010canard,schurov2017duck}. For the RSJ model, the slow-fast dynamics was analyzed in~\cite{kleptsyn2013josephson}. The interpretation in terms of an auxiliary Riemann surface allows us to get the non-perturbative contributions in the dynamical system in the leading order. In all cases, the summation of the non-perturbative effects involves the auxiliary Riemann surface with the additional ingredients. This Riemann surface is identified with the complexified slow manifold in the slow-fast system.

This study combines different concepts from the theory of dynamical systems, differential equations, and quantum mechanics. We tried our best to provide the reader with a detailed description of each object, starting from the definitions to computation of essential quantities. For the cohesive narrative, we recall important results that were obtained earlier. Despite the seeming incoherency, during the study, we use very similar ideas, and the monodromy of the Hill equation is one of our cornerstones. The goal of this study is to track the fate of Poincar\'{e} rotation number as one goes from classical mechanics to quantum mechanics.

The link between the dynamical systems on the torus, potentially interpreted as the RG flow, and quantum mechanics with periodic potentials seems to be fruitful for both. The new approaches developed for exact quantization in quantum mechanics provide new insights for dynamical systems, while familiar results for the Poincar\'{e} rotation number and results in bifurcation theory in a more general context could be of some use in the quantum-mechanical context.

The remainder of the paper is organized as follows. In~\autoref{sec:Torus-Dyn} we introduce basic definitions and provide the reader with knowledge of how to determine phase-locking domains when the rotation number is quantized. We also explicitly derive the connection between the Hill equation and the dynamical system on the torus with the examples of the RSJ model and the Mathieu equation. Using the mentioned correspondence, in~\autoref{sec:Hannay-Rho} we show the connection between the nonadiabatic Hannay angle in the classical parametric oscillator and the rotation number of the dynamical system on the torus. In~\autoref{sec:Virasoro} we explain the interpretation of the phase-locking phenomenon in dynamical systems in terms of the classification of the coadjoint Virasoro orbits. Then, \autoref{sec:QM-Rho} and \autoref{sec:Slow-Fast-WKB} are devoted to the relation of the rotation number and phase-locking phenomenon with quantum mechanics with periodic potentials. The quantization of rotation number is linked to the Milne quantization approach, while the semiclassical and exact WKB approaches are related to the slow-fast behavior of dynamical systems. In~\autoref{sec:Discussion}   we collect an extensive list of questions for further study, while the summary of the results can be found in~\autoref{sec:Conclusion}. The technical details used in the main body of the paper are clarified in the Appendices.

\section{Dynamical systems on torus}\label{sec:Torus-Dyn}

\subsection{Flows on the torus}

Consider an arbitrary flow on 2-torus $T^2=S^1\times S^1=\mathbb{R}^2_{\phi,t}\slash2\pi\mathbb{Z}^2$ given by a non-autonomous differential equation of the type 
\begin{equation}\label{eq:torus-dyn-sys}
    \dot{\phi}=\frac{d\phi}{dt}=f(\phi,t)
\end{equation}
where $f(\phi,t)$ is a smooth function $2\pi$-periodic in each variable. Consider an arbitrary solution $\phi(t)$ of equation~\eqref{eq:torus-dyn-sys}, it is uniquely defined by its initial condition $\phi(t_0)=\phi_0$. Note that if $\phi(t)$ is a solution, then $\phi(t+2\pi)$ is also a solution. Recall that the {\it rotation number of flow}~\eqref{eq:torus-dyn-sys} is the limit
\begin{equation*}
    \rho=\lim\limits_{T\rightarrow\infty}\frac{\phi(T)}{T}\in\mathbb{R}
\end{equation*}
which is known to exist and to be independent on the initial condition $(\phi_0,t_0)$, see \cite[p. 104]{arn}.
\begin{figure}
    \centering
    \includegraphics{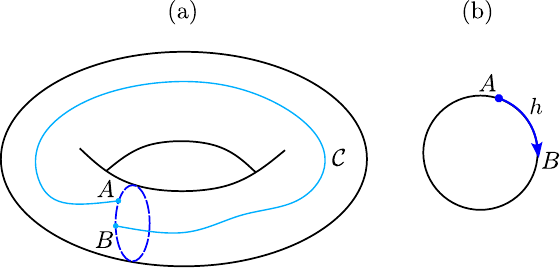}
    \caption{(a) Trajectory $\mathcal{C}$ on torus that start at point $A$ and after one period ends at point $B$, the dashed circle represents the Poincar\'{e} section, (b) Poincar\'{e} map $h$ that maps $A\in S^{1}$ to the point $B\in S^{1}$}
    \label{fig:dynamics-on-torus}
\end{figure}
Recall that the \textit{Poincar\'{e} map} of the flow~\eqref{eq:torus-dyn-sys} is the orientation-preserving diffeomorphism $h:S^1\to S^1$ of the circle $S^1=S^1_{\phi}\times\{0\}=\mathbb{R}\slash2\pi\mathbb{Z}$ defined as follows. Take a $\phi_0\in S^1$ and the solution $\phi(t)$ with the initial condition $\phi(0)=\phi_0$. By definition,
\begin{equation*}
    h(\phi_0)=\phi(2\pi).
\end{equation*}
Each orientation-preserving circle homeomorphism $h:S^1\to S^1$ 
can be written in the coordinate $\phi$ on the universal covering line $\mathbb{R}$ of the circle as
\begin{equation*}
    h(\phi)=\phi+g(\phi) \ \ g(\phi) \text{ is a continuous } 2\pi\text{-periodic function.}
\end{equation*}
The function $g(\phi)$ is uniquely defined up to addition of 
number $2\pi m$, $m\in\mathbb{Z}$. Recall that by definition, the {\it rotation number of a homeomorphism} $h$ introduced by H. Poincar\'{e}, see \cite[section 11]{arn}, is  equal to
\begin{equation*}
    \rho=\lim_{k\rightarrow+\infty}\frac{h^k(\phi_0)}{2\pi k}(\mathrm{mod}\,\,\mathbb{Z}).
\end{equation*}
The limit is known to be independent on $\phi_0$. As we add a number $2\pi m$ to $g(\phi)$, this results in adding $m\in\mathbb{Z}$ to $\rho$, which explains why $\rho$ is a well-defined modulo $\mathbb{Z}$. 

It is well known that the rotation number of the flow~\eqref{eq:torus-dyn-sys} taken modulo $\mathbb{Z}$ is equal to the rotation number of its Poincar\'{e} map.

\subsection{Hill equation as dynamical system}\label{subsec:DynTorusHill}

Hill equation is the second order differential equation with periodic coefficients of form
\begin{equation}
    \ddot{u} + g_1(\tau)\dot{u} + g_2(\tau)u=0,\label{hillg}
\end{equation}
where $g_1(\tau)$ and $g_2(\tau)$ are real-valued $2\pi$-periodic functions and time was rescaled in order to have $2\pi$-periodicity, dot denotes derivative with respect to $\tau$. It is known that the number of zeros $N(T)$ of a solution $u(\tau)$ of equation (\ref{hillg}) on the interval $[0,T)$ is related to the rotation number $\rho$ of certain dynamical system on torus \cite{johnson1982rotation},
\begin{equation}\label{eq:Hill-rho-connection}
    \rho = \pi\lim\limits_{T\rightarrow\infty}\frac{N(T)}{T}.
\end{equation}
The most straightforward way to prove the statement is the application of Pr\"{u}fer transformation~\cite{prufer1926neue}. For the sake of completeness, we provide the proof in \autoref{app:Hill-Rho}.

Pr\"{u}fer transform is nothing more than the standard radial projection of the Hill equation to the unit circle. The Hill equation be considered as a first order linear differential equation 
\begin{equation}\label{eq:Hill-LinearSys}
    \begin{cases} 
    \dot{u} = v\\
    \dot{v} = -g_2(\tau)u-g_1(\tau)v
    \end{cases}
\end{equation}
on the vector function $(u(\tau),\,v(\tau))=(u(\tau),\,\dot{u}(\tau))$. The non-autonomous flow map of system~\eqref{eq:Hill-LinearSys} is given by a fundamental matrix solution
\begin{equation*}
    \tau\mapsto U(\tau)\in \mathrm{GL}(2,\mathbb{R}): \ \text{ the columns of the matrix } U(\tau) \text{ are basic solutions of~\eqref{eq:Hill-LinearSys}.}
\end{equation*}
The standard radial projection to the unit circle, post-composed with symmetry
\begin{equation*}
    \mathbb{R}^2_{u,v}\to S^1_{u,v}:=\{|u|^2+|v|^2=1\}, \ \ (u,v)\mapsto(\frac u{\sqrt{u^2+v^2}}, -\frac v{\sqrt{u^2+v^2}}),
\end{equation*}
sends equation~\eqref{eq:Hill-LinearSys} to a differential equation (line field) on the torus $T^2=S^1_{u,-v}\times S^1_\tau$. The coordinate on the first circle is 
\begin{equation*}
    \xi=\arg(u-iv)=-\arctan\frac vu.
\end{equation*}
The differential equation on $T^2$ thus obtained by projection of linear system~\eqref{eq:Hill-LinearSys}is 
\begin{equation}\label{eq:Hill-Torus-DynSys}
    \frac{d\xi}{d\tau}=\frac{g_2(\tau)+1}2+\frac{g_2(\tau)-1}2\cos2\xi-\frac{g_1(\tau)}2\sin2\xi.
\end{equation}
The relation~\eqref{eq:Hill-rho-connection} can be represented in other forms, which are briefly listed in \cite{johnson1982rotation}. Since we will use some of these forms later on, we also list these relations explicitly. Keeping in mind that any second order differential equation can be transformed into its normal form, we can consider the following spectral problem,
\begin{equation}\label{eq:dyn-sys-spectral}
    \left\{-\frac{d^2}{d\tau^2} + q(\tau)\right\}u(\tau)=\lambda u(\tau),
\end{equation}
where $q(\tau)$ can be quasi-periodic function and $\lambda$ is the eigenvalue of the differential operator in the left hand side. For this equation, in addition to the relation~\eqref{eq:Hill-rho-connection} the rotation number can be defined as
\begin{equation}
    \rho(\lambda) = \lim\limits_{T\rightarrow\infty}\frac{\arg(\dot{u}(T;\lambda)+iu(T;\lambda))}{T}.
\end{equation}
By the Sturm's comparison lemma, the number of eigenvalues that are bounded by a fixed value $\lambda$ on the interval $[0,T]$, i.e. the density of states $\nu(T;\lambda)$ coincides up to addition of $\pm 1$ with the number of solution zeros $N(T;\lambda)$, so the rotation number is related to the density of states,
\begin{equation}
    \rho(\lambda)=\pi\lim\limits_{T\rightarrow\infty}\frac{\nu(T;\lambda)}{T}.
\end{equation}
The relation between rotation number and density of states in the spectral problem of type~\eqref{eq:dyn-sys-spectral} was studied intensively in~\cite{broer1995geometrical,broer2000resonance}.

\subsection{Josephson junction as illustrative example}

As we have already stated in the Introduction, the Josephson junction (JJ) is one of the simplest but important examples of dynamical system on the torus with rotation number quantization effect. The simplest possible theory that describes JJ is the so-called RCSJ model, described in~\cite{mccumber1968effect} and investigated in detail in~\cite{schon1990quantum,zaikin2019dissipative}. In this model, the phase difference between two superconductors is governed by the following equation,
\begin{equation}\label{eq:RSCJ-model}
    \frac{C\hbar}{2e}\ddot{\phi}+\frac{\hbar}{2eR}\dot{\phi}+I_c\sin\phi=I(t),
\end{equation}
where $I=I(t)$ is the external current, $R$ is the shunting resistance, $e$ is the electron charge. It is convenient to introduce dimensionless time, $\tau=t/\tau_J$, where $\tau_J=\hbar/(2eRI_c)$ is the so-called Josephson time. After this change of variables, the eq.~\eqref{eq:RSCJ-model} becomes
\begin{equation}
    \epsilon\frac{d^2\phi}{d\tau^2}+\frac{d\phi}{d\tau}+\sin\phi=\frac{I}{I_c},\quad \epsilon=2CR^2I_ce\hbar^{-1},
\end{equation}
where $\epsilon$ is dimensionless Stewart-McCumber parameter. The overdamped limit corresponds to $\epsilon\ll 1$, so we obtain RSJ model,
\begin{equation}\label{eq:RSJ model}
    \frac{d\phi}{d\tau}+\sin\phi=\frac{I}{I_c}.
\end{equation}
The time-averaged value of voltage on the junction is given by the Josephson relation,
\begin{equation}\label{eq:JJ-voltage}
    \langle V\rangle = \lim\limits_{T\rightarrow\infty}\frac{1}{T}\int_{0}^{T}dt\,V(t)=\lim\limits_{T\rightarrow\infty}\frac{1}{T}\int_{0}^{T}dt\left(\frac{\hbar}{2e}\frac{d\phi}{dt}\right)=\frac{\hbar}{2e}\lim\limits_{T\rightarrow\infty}\frac{\phi(T)-\phi(0)}{T}.
\end{equation}
We see that the voltage on JJ coincides with the definition of the rotation number up to the factor $\hbar/(2e)$. The current-voltage characteristic has a staircase structure, known as \emph{Shapiro steps}. The question was: does this current-voltage graph have only integer steps or non-integer steps are also possible? For an external current of the form $I/I_c=B+A\cos\omega t$ ($B$ is the DC component, $A\cos\omega t$ is the AC component), direct simulations show that the steps are integer only; see Fig.~\ref{fig:RSJ-voltage}.
\begin{figure}
    \centering
    \includegraphics[width=0.3\linewidth]{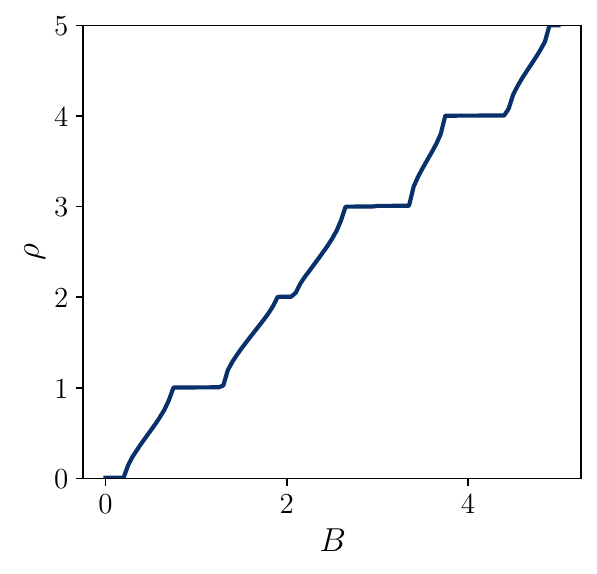}
    \caption{Rotation number $\rho$ as the function of $B$ for fixed value of $A$ for the RSJ model.}
    \label{fig:RSJ-voltage}
\end{figure}
The mathematically rigorous proof that the steps are integer only was made in \cite{buchstaber2010rotation} and is based on the quantization effect of the rotation number. However, the origin of Shapiro steps were derived many years before the first papers related to rotation number quantization. Of course, on the physical level of rigor, without accurate proofs, and so on. It seems that almost the first paper related to the existence of only integer Shapiro steps is \cite{renne1974some}. This paper contains an unnoticeable and underrated statement that connects the RSJ model with the Hill equation and then simplifies the demonstration of integer-only steps. After several changes of variables, the RSJ model equation can be written as the Hill equation,
\begin{equation}\label{eq:RSJ-Hill}
    \ddot{u}+V(t)u=0,
\end{equation}
with the following potential function,
\begin{equation}\label{eq:RSJ-Hill-potential}
    V(t)=-\frac{1}{4}+\frac{1}{4}\left(B+A\cos\omega t\right)^2-\frac{1}{4}\left(\frac{2A\omega^2\cos\omega t}{1+B+A\cos\omega t}+\frac{3A^2\omega^2\sin^2\omega t}{(1+B+A\cos\omega t)^2}\right),
\end{equation}
which has period $T=2\pi/\omega$. For the sake of completeness, we repeat all relevant derivations from paper \cite{renne1974some} in~\autoref{app:Hill-RSJ}. The Floquet theory allows us to investigate the corresponding Hill equation in great detail. The most interesting question for us is to reproduce the phase-locking areas in the RSJ model from the stability chart of the Hill equation. That can be done by computing the monodromy matrix.

Let $u_1(t)$ \& $u_2(t)$ be two independent solutions of the Hill equation~\eqref{eq:Hill}. The monodromy matrix is defined by the relation
\begin{equation}
    \left(u_1(t+T)\,\,\,u_2(t+T)\right)^T = M\,\left(u_1(t)\,\,\,u_2(t)\right)^T
\end{equation}
where the matrix $M\in\mathrm{SL}(2,\mathbb{R})$. According to Floquet theory, the Hill equation has stable (in sense that they are bounded) solutions if the monodromy matrix is elliptic and has unstable (unbounded) solutions if the monodromy matrix is hyperbolic. The notation ``elliptic'' and ``hyperbolic'' are related to the classification of $\mathrm{SL}(2,\mathbb{R})$ matrices by conjugacy classes. This classification is quite simple: $M\in\mathrm{SL}(2,\mathbb{R})$ is called elliptic if $|\tr M| < 2$, hyperbolic if $|\tr M| > 2$, and parabolic if $|\tr M| = 2$. All these matrices are characterized by two eigenvalues, $\lambda_1$ and $\lambda_2$, such that $\tr M = \lambda_1 + \lambda_2$ and $\det M = \lambda_1\lambda_2=1$. In Floquet theory the quantities $\lambda_{1,2}$ are related to the \emph{Floquet exponents} $\nu_{1,2}$ by the relation $\lambda_{1,2}=\exp(\nu_{1,2}T)$. The real parts of the Floquet exponents are nothing more than the Lyapunov exponents. Obviously, the monodromy matrix and, as a consequence, the Floquet exponents depend on parameters of the potential $V(t)$.

The most straightforward way to determine the stability and instability domains in the parameter space is by solving the Hill equation numerically followed by calculating the monodromy matrix and its trace. This can be achieved in many different ways, but we recall the implicit method described in the recent paper~\cite{chikmagalur2024implicit}. We follow a fairly standard approach: for each point $(B,A)$ we numerically solve the Hill equation with the same initial conditions and then compute the trace of a monodromy matrix.
\begin{figure}
    \centering
    \includegraphics[width=\linewidth]{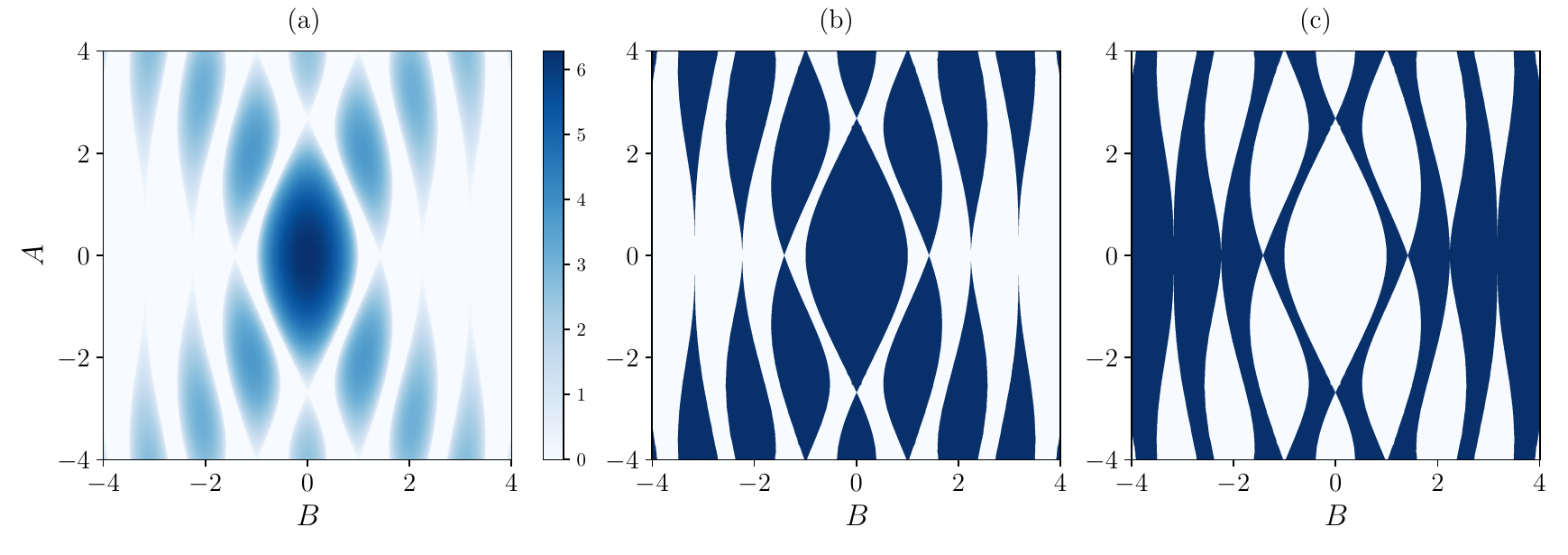}
    \caption{(a) largest Lyapunov exponent chart for the RSJ model, (b) phase-locking domains (shaded by dark blue) in the RSJ model, (c) band structure of Hill equation corresponding to RSJ model (stability zones shaded by dark blue)}
    \label{fig:RSJ-Phase-Locking}
\end{figure}
For the RSJ model, the potential $V(t)$ in eq.~\eqref{eq:RSJ-Hill-potential} has poles if $A\geq|B+1|$. This means that integration over the period should be performed with care. The described numerical computation gives us the following plot of the band structure for the Hill equation corresponding to RSJ model, see Fig.~\ref{fig:RSJ-Phase-Locking}. Comparing this plot with numerical simulations and experimental data in \cite{panghotra2020giant,karpov2008modeling}, one can see that the phase-locking domains correspond to the \emph{instability} domains of the Hill equation, whereas the \emph{stability} domains correspond to areas without phase-locking. This can be easily explained with the help of the theorem that relates the zeros of Hill equation solutions and $\mathrm{SL}(2,\mathbb{R})$ conjugacy classes. We provide an explanation of the correspondence between phase-locking domains and instability domains in~\autoref{app:Layp-Rho}.

\subsection{Mathieu equation as simplest example}

One more example of correspondence between Hill equation \& dynamical system on torus is the Mathieu equation,
\begin{equation}\label{eq:Mathieu-system}
    \ddot{u}+\left(B+A\cos\omega t\right)u=0.
\end{equation}
The corresponding dynamical system on torus is
\begin{equation}
    \dot{\xi} = \sin^2\xi + \left(B+A\cos\omega t\right)\cos^2\xi
\end{equation}
which can be rewritten in more convenient form by setting $\phi=2\xi$,
\begin{equation}
    \dot{\phi} = (B+A\cos\omega t+1)+(B+A\cos\omega t-1)\cos\phi.
\end{equation}
The study of the Mathieu equation band structure attracts a lot of attention from both physicists and mathematicians. The phase-locking domains and the staircase structure for the Mathieu equation are represented in Fig.~\ref{fig:Mathieu-Plots}. Since the very simple form of potential, the perturbation theory with respect to non-linearity is well developed and many analytical results are known. Further, we often illustrate our findings with the help of Mathieu equation and the RSJ model.
\begin{figure}
    \centering
    \includegraphics[width=\linewidth]{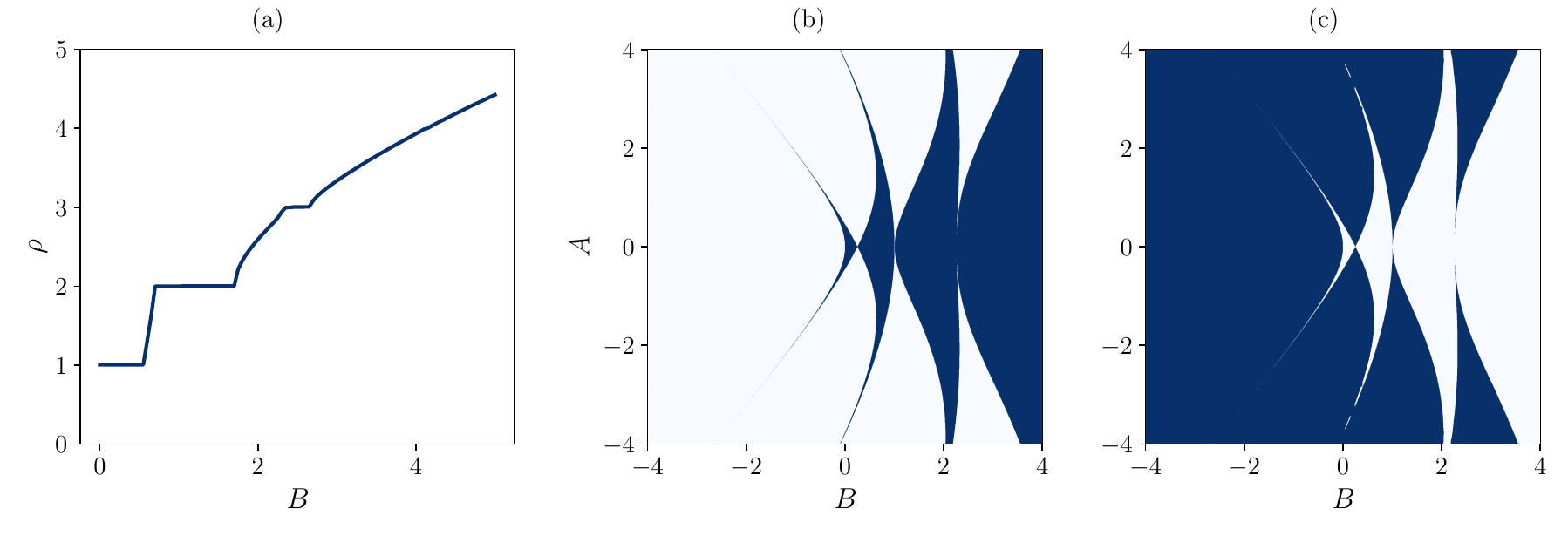}
    \caption{(a) rotation number quantization for Mathieu system ($A=2$, $\omega=1$) (b) band structure of Mathieu equation (stability zones shaded by dark blue), (c) phase-locking domains (shaded by dark blue) in the Mathieu system}
    \label{fig:Mathieu-Plots}
\end{figure}
In conclusion of this section, let us reinforce the correspondence between the Hill equation and the dynamical system on the torus in Tab.~\ref{tab:Hill-Dyn-corr}. From this table, one can notice that in the case of hyperbolic monodromy the relation between the Hill equation and the rotation number is very nice and simple. We argue that in the elliptic case the rotation number still has an interesting interpretation in terms of the non-adiabatic Hannay angle.

\begin{table}[h!]
\centering
\begin{tabular}{|c|c|c|}
\hline
$\mathrm{SL}(2,\mathbb{R})$ element & Hill equation            & Dynamical system on torus                 \\ \hline
---                                 & number of solution zeros $n$ & rotation number $\rho=2n$                           \\ \hline
hyperbolic                       & instability zone         & phase-locking domain, $\rho\in\mathbb{Z}_{\geq0}$ \\ \hline
elliptic                            & stability zone           & no phase-locking, $\rho\notin\mathbb{Z}$  \\ \hline
parabolic                           & boundary between zones   & boundary of phase-locking domain          \\ \hline
\end{tabular}
\caption{Correspondence between phase-locking domains and stability zones}
\label{tab:Hill-Dyn-corr}
\end{table}

\section{Hannay angle and rotation number}\label{sec:Hannay-Rho}

Let us remind that we have started from the Hill equation,
\begin{equation}\label{eq:Hill}
    \ddot{u}+V(t)u=0,
\end{equation}
where $V=V(t)$ is the periodic function of period $T$. As was shown above, projectivization onto the unit circle allows us to define the phase $\xi(t)$,
\begin{equation}\label{eq:RadialProjection}
    \xi(t) = \arg (u-i\dot{u}),
\end{equation}
where $u=u(t)$ is the solution of Hill equation. The phase $\phi(t)=2\xi(t)$ defines the dynamical system on torus,
\begin{equation}\label{eq:Hill-DynTorus}
    \dot{\phi}=(V(t)+1)+(V(t)-1)\cos\phi(t),
\end{equation}
which has the following rotation number,
\begin{equation}
    \rho = \lim\limits_{t\rightarrow\infty}\frac{\phi(t)-\phi(0)}{t}.
\end{equation}
From the previous section, one can easily identify the correspondence between the rotation number and the Hill equation solution. Here, our goal is to obtain a physical interpretation of the rotation number in the case of the elliptic Hill equation.

\subsection{Definition of non-adiabatic Hannay angle}

From a physical point of view, equation~\eqref{eq:Hill} corresponds to the time-dependent Hamiltonian system. Indeed, without loss of generality, we can associate the Hill equation~\eqref{eq:Hill} with the following Hamiltonian system,
\begin{equation}\label{eq:Hill-Hamiltonian}
    H(p,u;\,V(t)) = \frac{p^2}{2} + \frac{V(t)u^2}{2},
\end{equation}
where $p$ is the momentum, $u$ is the coordinate and we have emphasized that the Hamiltonian depends on the parameter $V(t)$ and it is the periodic function of time with period $T$. This is nothing more than the harmonic oscillator with time-dependent frequency, a.k.a. \emph{parametric oscillator}. The equation of motion for the coordinate $u$ coincides with the Hill equation~\eqref{eq:Hill}. For the Hamiltonian systems, one typically desires to find the action-angle variables, that make the Hamilton equations looks very simple,
\begin{equation}\label{eq:acition-angle}
    \dot{\theta} = \frac{\partial H(I,\theta)}{\partial I},\quad \dot{I}=0.
\end{equation}
When the Hamiltonian is time-dependent but this dependence is periodic, it was shown by Berry \& Hannay in~\cite{berry1988classical} that during the evolution over period the angle $\theta$ obtains the following change,
\begin{equation}
    \Delta\theta =  \Delta \theta_{d} + \theta_H.
\end{equation}
The first term is the so-called \emph{dynamical} phase, while the second term is the so-called \emph{non-adiabatic Hannay angle} \cite{aharonov1987phase,berry1988classical}. Of course, this quantity is well-defined only if the system trajectories are bounded. It is pretty reasonable to ask: How is the evolution of the angle variable $\theta$ related to the rotation number $\rho$ and the phase $\phi(t)$? To answer this question, we first give the definition of the Hannay angle, and in the next subsection we derive the explicit relation between the rotation number and the Hannay angle.

As a system with time-dependent Hamiltonian, the parametric oscillator is special in some sense. It was shown that this system has the \emph{exact} invariant \cite{lewis1968class}, which is given by
\begin{equation}
    I(t) = \frac{1}{2}\left[\frac{u^2}{w^2}+\left(u\dot{w}-w\dot{u}\right)^2\right]
\end{equation}
and it is called Lewis-Ermakov invariant. The function $w=w(t)$ solves the so-called Ermakov-Pinney (more correctly, Steen-Ermakov-Pinney equation and hereafter we denote this equation as EP),
\begin{equation}
    \ddot{w}+V(t)w = w^{-3}.
\end{equation}
Direct derivations show that $I$ is related to the Wronskian of two independent Hill equation~\eqref{eq:Hill} solutions $u_1$ \& $u_2$ as $I=W^2/2$. Existence of such invariant allows one to construct the pair of conjugated variables, $I$ and $\theta$, which can be expressed in terms of original variables $p=\dot{u}$ and $u$ as follows,
\begin{equation}\label{eq:Lewis-Action-Angle}
    I=\frac{1}{2}\left[\frac{u^2}{w^2}+\left(u\dot{w}-w\dot{u}\right)^2\right],\quad \theta = -\arctan\left[w^2\left(\frac{\dot{u}}{u}-\frac{\dot{w}}{w}\right)\right],\quad \{\theta,\,I\}=1.
\end{equation}
Therefore, despite the time dependence of the Hamiltonian, we can still find action-angle variables in explicit form. Then, we can compute the non-adiabatic Hannay angle~\cite{berry1988classical},
\begin{equation}
    \theta_H = \frac{1}{2\pi}\int_{0}^{2\pi}d\theta\int_{0}^{T}dt\frac{d}{dt}\left(\frac{\partial F_2}{\partial I}\right),
\end{equation}
where $F_2$ is generating function. This generating function $F_2$ is related to the angle variable as $\theta=\partial F_2/\partial I$, then
\begin{equation}\label{eq:Hannay-angle}
    \theta_H = \frac{2\pi}{2\pi}\int_{0}^{T}\frac{dt}{w(t)^2}=\int_{0}^{T}\frac{dt}{w(t)^2}
\end{equation}
The Hannay angle contains the information concerning  EP equation solution $w=w(t)$.

The last point to mention is that there exists a very simple answer to why the Lewis invariant exists. In the paper~\cite{eliezer1976note} it was shown that this invariant corresponds to the squared angular momentum in the auxiliary 2D problem. This problem describes the motion of the particle in the $V(t)$ potential in the 2D plane. The equation of motions are
\begin{equation*}
    \ddot{\bm{x}}+V(t)\bm{x}=0,
\end{equation*}
where $\bm{x}=(x,y)^T$. Using the ansatz $x(t)=w(t)\sin\theta(t)$, $y(t)=w(t)\sin\theta(t)$, it is straightforward to verify that the Lewis invariant looks like $I=L^2/2$, where $L$ is the angular momentum (see eq.~\eqref{eq:2D-plane} and~\eqref{eq:Lewis-as-Momentum} in the next subsection for intuition).

\subsection{Relation between Hannay angle and rotation number}

For given independent real solutions, say $u_1$ \& $u_2$, of Hill equation~\eqref{eq:Hill} the solution of EP equation can be written as
\begin{equation}\label{eq:EP-general-solution}
    w(t)=\sqrt{au_1^2+2bu_1u_2+cu_2^2},\quad ac-b^2=W^{-2},
\end{equation}
where $W$ is the Wronskian \cite{pinney1950nonlinear}. Here $a$, $b$, $c$ are arbitrary numbers satisfying the relation $ac-b^2=W^{-2}$. Let us consider the elliptic case and write down these independent solutions in form
\begin{equation}\label{eq:EP-elliptic}
    u_1(t)=R(t)\cos\theta(t),\quad u_2(t)=R(t)\sin\theta(t).
\end{equation}
Next, we introduce the complex-valued function $u_1(t)+iu_2(t)=R(t)\exp\{i\theta(t)\}$. Differentiating twice with respect to time, we find the equations for $R=R(t)$ and $\theta=\theta(t)$,
\begin{equation}\label{eq:2D-plane}
    \ddot{R}+V(t)R-R\dot{\theta}^2=0,\quad R\ddot{\theta}+2\dot{R}\dot{\theta}=0.
\end{equation}
We can integrate the second equation
\begin{equation}\label{eq:Lewis-as-Momentum}
    \dot{\theta}=WR^{-2},\quad c=\dot{\theta}(0)R^2(0)\neq 0.
\end{equation}
One can easily check that $W=u_1(0)\dot{u}_2(0)-u_2(0)\dot{u}_1(0)\neq 0$, that is, it is nothing more than the Wronskian of two solutions. Then, we substitute $\dot{\theta}=WR^{-2}$ into the first equation and find
\begin{equation}
    \ddot{R}+V(t)R=W^2R^{-3}.
\end{equation}
By rescaling $w(t)=|W|^{-1/2}R(t)$, we see that $w(t)$ obeys the EP equation,
\begin{equation}
    \ddot{w}+V(t)w=w^{-3}.
\end{equation}
Conversely, let $w=w(t)$ be a positive solution of the EP equation. Then, keeping in mind the relation between $R(t)$ and $w(t)$, we can write the following expression for the solution of the equation for $\theta=\theta(t)$ of system~\eqref{eq:Lewis-as-Momentum},
\begin{equation}
    \theta(t) = \int_{0}^{t}\frac{ds}{w(s)^2}.
\end{equation}
One has $w(t)^2=u_1(t)^2+u_2(t)^2$ for appropriately chosen independent solutions $u_1$, $u_2$ with unit Wronskian, which follows from (\ref{eq:EP-general-solution}).
When the Hill equation~\eqref{eq:Hill} has an elliptic monodromy, the function $w(t)$ cannot be vanished to zero. This relation leads to the following: the rotation number $\rho$ that corresponds to the dynamical system on the torus (induced by projectivization of the Hill equation into the unit circle) is expressed by the Hannay angle and given by
\begin{equation}
    \rho = 2\,\lim\limits_{t\rightarrow\infty}\frac{\theta(t)}{t}=\lim\limits_{t\rightarrow\infty}\frac{2}{t}\int_{0}^{t}\frac{ds}{u_1(s)^2+u_2(s)^2}=\lim\limits_{t\rightarrow\infty}\frac{2}{t}\frac{t\theta_H}{T}=\frac{2\theta_H}{T},
\end{equation}
where $u_1$ and $u_2$ are two independent solutions of the Hill equation with unit Wronskian; $\theta_H$ is the non-adiabatic Hannay angle, defined in~\eqref{eq:Hannay-angle} and $T$ is the period of potential $V(t)$.

It happens that in the case of hyperbolic and parabolic monodromies of the Hill equation there are very similar statements that connect the rotation number of the corresponding dynamical system on the torus and combination of the Hill equation solutions. These statements are consequences of the classification of coadjoint orbits of the Virasoro group based on the properties of the Hill operator.

\section{Hill operator and coadjoint Virasoro orbits}\label{sec:Virasoro}

This section is devoted to the discussion of the relation between rotation number and classification of coadjoint Virasoro orbits. As a first hint, we use the proposed connection between the rotation number and the Hannay angle.

\subsection{Classification of coadjoint orbits}

In previous sections, we have discussed the relation between Hill equation and the dynamical system on 2D torus. It was emphasized that the properties of these systems in the torus are directly related to circle diffeomorphisms $\mathrm{Diff}(S^1)$. Let us recall some notations concerning $\mathrm{Diff}(S^1)$. The Lie algebra $\mathrm{Diff}(S^1)$ involves the vector field $f(t)\frac{d}{dt}$, while a co-adjoint vector is quadratic differential $V(t)(dt)^2$ (without loss of generality here we assume that the potential has the period $2\pi$). The pairing in the Lie algebra is defined as
\begin{equation*}
    \left(f,\,V\right)=\int_{0}^{2\pi}dt\,f(t)V(t).
\end{equation*}
We will consider the Virasoro group, which is the universal central extension $\widehat{\mathrm{Diff}(S^1)}$ of $\mathrm{Diff}(S^1)$. Now, the
Lie algebra involves the vector fields together with the central element $c$ obeying the commutation relations
\begin{equation*}
    \left[L_n,L_m\right]= (m-n) L_{n+m} +\frac{c}{12}m^3\delta_{n+m,0},
\end{equation*}
where $L_m=ie^{imt}\frac{d}{dt}$. The coadjoint element is 2-differential with the central extension $q$, which is identified as $q=c$. The quadratic differential transforms under the action of Virasoro group as
\begin{equation}
    \delta V= 2f'V+fV' - \frac{f'''q}{24 \pi}, \qquad \delta q=0.
\end{equation}
The vector field $f(t)\frac{d}{dt}$ is called stabilizer if and only if the function $f(t)$ obeys the following equation,
\begin{equation}\label{eq:stab-Vir}
    \frac{1}{2}f'''(t)+2V(t)f'(t)+V'(t)f(t)=0,
\end{equation}
where $f'=df/dt$ and maintains the 2-differential $V$ intact.
The eq.~\eqref{eq:stab-Vir} can be written in another form,
\begin{equation}
    \frac{1}{2}f(t)f''(t) + V(t)f(t)^2-\frac{1}{4}(f'(t))^2=C,
\end{equation}
where $C$ is the constant. Setting $w(t)=\sqrt{f(t)}$, we obtain
\begin{equation}\label{eq:Vir-EP}
    w''(t)+V(t)w(t)=Cw(t)^{-3},
\end{equation}
which is nothing more than the EP equation and without loss of generality we can set hereafter $C=1$. The stabilizers $f\in\mathrm{Stab}_V$ are in the isotropy subgroup of the Hill operator $\partial_t^2+V(t)$. As we have already mentioned, any solution $w(t)$ of eq.~\eqref{eq:Vir-EP} can be constructed by finding two independent solutions, say $u_1(t)$ and $u_2(t)$, of the Hill equation. This implies that the stabilizer can also be expressed in terms of the combination of two solutions $u_1(t)$ and $u_2(t)$.

The classification of coadjoint orbits is closely related to the classification of the $\mathrm{SL}(2,\mathbb{R})$ group by conjugacy classes, which is discussed in detail in \cite{lazutkin1975normal,kirillov1981orbits,witten1988coadjoint,balog1998coadjoint}. The main statement is that each coadjoint orbit is characterized by two quantities, $\Delta$ and $n$ (we clarify their definitions and properties further).

Since the coadjoint orbits of the Virasoro group are given by $\mathrm{Diff}(S^1)/\mathrm{Stab}_V$, they can be classified by the vector field $f(t)$, which obeys the eq.~\eqref{eq:stab-Vir}. Based on the conjugacy classes of $\mathrm{SL}(2,\mathbb{R})$, the stabilizer $f(t)$ can be represented via the following combinations of two independent real solutions of Hill equation $u_1(t)$ and $u_2(t)$ with unit Wronskian, $W=1$,
\begin{equation*}
    f_{\text{ell}}(t)=u_1(t)^2+u_2(t)^2, \quad f_{\text{hyp}}(t)=2u_1(t)u_2(t),\quad f_{\text{par}}(t)=u_1(t)^2\,\text{or}\,u_2(t)^2.
\end{equation*}
Each coadjoint orbit of the Virasoro group can be classified by the number of zeros $n$ of the stabilizer vector field $f(t)$. There exists the additional invariant $\Delta$, defined by the relation
\begin{equation}
    \Delta = \int_{0}^{2\pi}\frac{ds}{f(s)}.
\end{equation}
So, the complete classification of the coadjoint orbits of the Virasoro group can be given in terms of two quantities, $n$ and $\Delta$. This full classification reads as
\begin{equation}
    \frac{\mathrm{Diff}(S^1)}{S^1},\qquad \frac{\mathrm{Diff}(S^1)}{\mathrm{SL}^n(2,\mathbb{R})}, \qquad T_{\Delta,n}, \qquad T_{\pm,n}
\end{equation}
where two parameters of the orbits are defined above. The details of the classification can be found, for example, in \cite{blau2024}. 

It is useful to provide the so-called \emph{representative} Hill potential $V_0(t)$ for each class of the Virasoro coadjoint orbit since that an arbitrary Hill potential $V(t)$ for a given class can be obtained from the representative potential $V_0(t)$ by the action $F\in\mathrm{Diff}(S^1)$,
\begin{equation}
    V(t) = V_0(t)\left(F'(t)\right)^2+\frac{1}{2}S(F(t)),\quad S(F)=\frac{F'''(t)}{F'(t)}-\frac{3}{2}\left(\frac{F''(t)}{F'(t)}\right)^2.
\end{equation}
For the simplest orbits of type $\mathrm{Diff}(S^1)/S^1$ the representative Hill potential looks like $V_0(t)=(\alpha+n)^2/4$, where $\alpha\in\mathbb{R}$ and $n\in\mathbb{Z}$. For $\alpha=0$, the potential is simply $V_0(t)=n^2/4$ and the stabilizer is generated by the vector fields $L_n$, $L_0$, and $L_{-n}$ on $S^1$, which correspond to orbits $\mathrm{Diff}(S^1)/\mathrm{SL}^n(2,\mathbb{R})$. 

There are special Virasoro orbits that correspond to vector fields with zeros that are deformations of $\mathrm{Diff}(S^1)/\mathrm{SL}^n(2,\mathbb{R})$. For the hyperbolic monodromy of Hill equation and the winding number $n$ these orbits are denoted by $T_{\Delta,n}$. They have stabilizer of form
\begin{equation*}
    f(t) = \frac{\sin nt}{n+\Delta \sin nt}
\end{equation*}
with $2n$ single zeros and typical diffeomorphism $F_{\Delta,n}(t)=e^{\Delta t}\tan(nt/2)$. Their non-constant representative Hill potential is
\begin{equation}
    V_0(t)= -\frac{\Delta^2}{4} -\frac{1}{2} \frac{n(n^2-\Delta^2)}{n+\Delta\sin(nt)} + \frac{3}{4} \frac{n^2(n^2-\Delta^2)}{(n+\Delta\sin(nt))^2}.
\end{equation}
The case of parabolic monodromy with $n$ number of zeros corresponds to $T_{\pm,n}$ orbits with stabilizer
\begin{equation}
    f(t)=\frac{2\cos^2(nt/2)}{n\pm \pi^{-1}\cos^2(nt/2)},
\end{equation}
which has $2n$ double zeros and typical diffeomorphism $F_{\pm,n}(t)=\pm(t/2\pi)+\tanh(nt/2)$. The representative Hill potential is given by
\begin{equation}
    V_0(t) = \frac{n^3}{8}\left\{\frac{6(n\pm \pi^{-1})}{(n\pm \pi^{-1}\cos^2(nt/2))^2}-\frac{8}{(n\pm \pi^{-1}\cos^2(nt/2))^2}\right\}.
\end{equation}

\subsection{Coadjoint orbits in RSJ model}

Consider the potential of the RSJ model $V(t)$ with $\omega=1$ (without loss of generality), which is given by~\eqref{eq:RSJ-Hill-potential}. For $A=0$, the potential becomes simply
\begin{equation*}
    V(t) = \frac{B^2-1}{4}.
\end{equation*}
The condition $B^2-1=n^2$, $n\in\mathbb{Z}$ determines the roots of Arnold tongues and tells that this potential lies on the orbit $T_{0,n}$. The interior of the phase-locking domains parameterized by two parameters $A$ and $B$. The corresponding Hill equation monodromy is hyperbolic. Hence the potential $V(t)$ in this domain lies on the $T_{\Delta,n}$ coadjoint orbit. On the boundary of phase-locking domains, the monodromy is parabolic, so the potential $V(t)$ lies on the $T_{\pm,n}$ coadjoint orbit. Between the phase-locking domains, the monodromy is elliptic. The corresponding the stabilizer does not have zeros. This means that this potential lies on the coadjoint orbit of type $T_{\alpha,0}$, $\alpha\in\mathbb{R}$.

Finally, the structure of phase-locking domains in the RSJ model contains quite interesting features: There exist points in the $(B,A)$ plane, where the Arnold tongues shrink to the point. These points are called \emph{constrictions}, and their properties were investigated in~\cite{glutsyuk2014adjacency}. For an arbitrary $\omega>0$, all constrictions are located in lines with $B=\omega k$, where $k$ is an integer, which is equal to the rotation number. For large values of $A$, the boundaries of Arnold tongues can be approximated by the Bessel functions $J_k(A)$, which is shown in \cite{klimenko2013asymptotic}. Using this fact, one can show that the period average of $V(t)$ in~\eqref{eq:RSJ-Hill} in $A\rightarrow\infty$ limit can be estimated by $j_{k,l}^2$, where $j_{k,l}$ is the $l$-th zero of $J_{k}(A)$. In \cite{buchstaber2017monodromy,kleptsyn2013josephson} it was shown that constrictions are characterized by degenerate (trivial) monodromy, they correspond to $T_{0,n}$ points.

Having stated the relation between coadjoint orbits and Hill potentials, we need to compute two quantities, $n$ and $\Delta$. For the number of zeros, we already know that it is related to the rotation number. The quantity $\Delta$ requires more attention. The crucial fact is that the eigenvalues of the monodromy matrix in the elliptic and hyperbolic cases can be expressed by the quantity $\Delta$ \cite{unterberger2010}. For the elliptic case, the eigenvalues $\lambda_{\pm}$ of the monodromy matrix are given by
\begin{equation}
    \lambda_{\pm} = \exp\left(\pm i\Delta\right),
\end{equation}
whereas for the hyperbolic case the eigenvalues are
\begin{equation}
    \lambda_{\pm} = \exp\left(\pm \Delta\right).
\end{equation}
For the sake of completeness, we provide the proof of these statements in \autoref{app:Inv-Virasoro}. As we have stated above, in the elliptic case, the irrational rotation number $\rho$ is given by the integral. Comparing this expression with $\Delta$, we conclude that the rotation number is related to eigenvalues of monodromy matrix,
\begin{equation}\label{eq:rotation-number-Delta}
    \rho = \frac{1}{2\pi i}\ln\left(\frac{\lambda_{+}}{\lambda_{-}}\right)=\frac{\Delta}{\pi}=\frac{1}{\pi}\int_{0}^{2\pi}\frac{dt}{u_1(t)^2+u_2(t)^2},
\end{equation}
which is consistent with consideration in \cite{buchstaber2017monodromy} (Proposition 5.6). Now we can improve this statement. First, the non-adiabatic Hannay angle is nothing more than the doubled rotation angle that corresponds to the elliptic monodromy matrix. In fact, in the elliptic case, the eigenvalues of the monodromy matrix are $\lambda_{\pm}=e^{\pm i\alpha}$. Comparing then the equation~\eqref{eq:Hannay-angle} and the expression for the invariant $\Delta$, we conclude that $\Delta=\alpha=\theta_H$, therefore $\rho=\alpha/\pi$ and Poincar\'{e} map corresponds to the rotation by the angle $2\alpha$, as should be.

Note that the structure of invariant combinations is directly related to the solutions of the EP equation. As we have mentioned earlier, the most general solution of the EP equation is given by~\eqref{eq:EP-general-solution}. From the structure of invariants, we see that the elliptic case corresponds to $B=0$, so $ac=W^{-2}$. So, we normalize Wronskian and then set $a=c=1$, so the solution of the EP equation becomes $w=u_1^2+u_2^2$ (cf. eq.~\eqref{eq:EP-elliptic}). In the hyperbolic case, we see that $a=c=0$ and $b=i$, so the solution is given by $w=2iu_1u_2$.

Using the RSJ model described by the Hill equation~\eqref{eq:RSJ-Hill} as an illustrative example, we conclude this section with the following table; see Tab.~\ref{tab:Virasoro-Rho}.
\begin{table}[h!]
\centering
\begin{tabular}{|c|c|c|}
\hline
Coadjoint orbit type & Dynamical system on torus       & Hill equation monodromy \\ \hline
$T_{\Delta,n}$       & phase-locking, $\rho=2n$         & hyperbolic              \\ \hline
$T_{0,n}$            & roots of Arnold tongues         & degenerate                 \\ \hline
$T_{\pm,n}$          & boundaries of Arnold tongues    & parabolic               \\ \hline
$T_{\alpha,0}$       & no phase-locking, $\rho=\alpha/\pi$ & elliptic                \\ \hline
\end{tabular}
\caption{Correspondence between Virasoro coadjoint orbits and the phase-locking}
\label{tab:Virasoro-Rho}
\end{table}

\subsection{Comments on the flows in the space of orbits}

Having identified the explicit relation between the geometry of the phase-locking in the RSJ model and the coadjoint orbits classification, we can make a few preliminary comments concerning the flows in the $(B,A)$-plane in the space of Virasoro orbits. There are natural flows along the orbits at fixed $(\Delta,n)$, as well as flows with varying $\Delta$ or $n$. Among the simplest flows are 
\begin{enumerate}
\setlength{\itemsep}{0pt}
\setlength{\parskip}{0pt}
    \item $A=0$ with varying $B$. This motion corresponds to the flow between the orbits
    \begin{equation*}
        T_{0,n} \rightarrow T_{-,n} \rightarrow T_{\alpha,0} \rightarrow T_{+,n+1} \rightarrow T_{0,n+1}
    \end{equation*}
    \item $A=\mathrm{const}$ with varying $B$. This motion corresponds to the flow between the orbits,
    \begin{equation*}
        T_{\Delta,n} \rightarrow T_{-,n} \rightarrow T_{\alpha,0} \rightarrow T_{+,n+1} \rightarrow T_{\Delta,n+1}
    \end{equation*}
    These two types of flow describe the structure of the Cantor staircase in the model. Note that the flow between the steps of the staircase corresponds to the elliptic orbit $T_{\alpha,0}$. This observation could be useful for the determination of the CFT for the flow between plateaus in the IQHE.
    
    \item $B=n$, $n\in\mathbb{Z}$ and varying $A$, 
    \begin{equation*}
        T_{0,n} \rightarrow T_{\Delta,n} \rightarrow T_{0,n} \rightarrow T_{\Delta,n} \rightarrow T_{0,n}
    \end{equation*}

    \item $B=\mathrm{const}\notin\mathbb{Z}$ and varying A. In a generic situation, the flow looks as
    \begin{equation*}
        T_{\alpha,0} \rightarrow T_{\pm,n} \rightarrow T_{\Delta,n} \rightarrow T_{\pm,n} \rightarrow T_{\alpha,0}
    \end{equation*}

\end{enumerate}
The flows along the orbits can be described in more quantitative terms. To this aim recall that Virasoro coadjoint orbits are symplectic manifolds with the canonical Kirillov-Kostant ``$p\,dq$'' geometrical action providing the Virasoro Poisson bracket. The geometric action for the simplest Virasoro orbits has been identified as the Liouville field theory \cite{alekseev1989path,wiegmann1989multivalued}. The Liouville field is defined through the diffeomorphism as $F'=e^{\phi}$
It can be obtained via the Drinfeld-Sokolov reduction from the geometric action on the coadjoint orbit of the $\hat{\mathrm{SL}(2,\mathbb{R})}$ Kac-Moody algebra- the WZW action. The relation with the $\hat{\mathrm{SL}(2,\mathbb{R})}$ Kac-Moody algebra provides an explanation of the winding number parameter of the Virasoro orbit since $\pi_1(\mathrm{SL}(2,\mathbb{R}))=\mathbb{Z}$ \cite{gorsky1991large}. The geometrical action corresponding to the $T_{\Delta,n}, T_{\pm,n}$ orbits involves additional perturbations; see, for example, the discussion in \cite{gorsky1995liouville,balog1998coadjoint}.

The equation of motion for the Liouville theory is $\dot{V}=0$ since it does not involve any Hamiltonian, $H=0$. To get the flows
on the coadjoint orbits, one has to add Hamiltonians. The family of integrable flows is provided by KdV Hamiltonians $H_k$
\begin{equation}
    \partial_{t_k} V=\{V,H_k\}
\end{equation}
with the generating function, expressed in terms of the trace of monodromy matrix, see, for instance \cite{bazhanov1996integrable} 
\begin{equation}
    \frac{1}{2\pi}\ln|\tr M(\lambda)|= \lambda - \sum_n c_n H_{2n-1} \lambda^{1-2n}
\end{equation}
where $c_1=1/2$, $c_n=(2n-3)!!/(2^nn!)$, $n>1$.The KdV Hamiltonians are in involution with respect to the Virasoro Poisson bracket $\{H_k,H_l\}=0$. Recall that the Hill operator is the Lax operator for the KdV hierarchy and the KdV evolution is isospectral, that is, $\lambda=\mathrm{const}$ if we consider the operator $\partial_x^2 + V(x) +\lambda$ with periodic $V(x)$. The isospectral KdV evolution of $V(x,t)$ goes along the Riemann surface defined as 
\begin{equation}
    \det\left(M(\lambda)-\eta\right)=0,
\end{equation}
where $M(\lambda)$ is monodromy matrix of the Schr\"{o}dinger operator and $\eta$ is the eigenvalue. 

For the Hill operator, the family of so-called finite-gap potentials can be extracted \cite{novikov1974periodic}, when only a finite number of KdV charges are independent. 
The example of such $N$-gap potential is $V(x)=N(N+1) \wp(x,\tau)$, where $\wp(x,\tau)$ is the Weierstrass elliptic function with modular parameter $\tau$. The generic Heun equation can be transformed into the Schr\"{o}dinger equation with the potential $V(x)=\sum_{i=1,2,3}l_i(l_i+1)\wp(x,\tau_i)$ which is finite-gap if $l_i\in N$. If we consider the Mathieu potential that corresponds to the certain limit of the Hill potential of the RSJ model, the isospectral evolution of KdV implies the evolution of the parameter $A$ at a fixed $B$, and the spectral curve has an infinite genus. In terms of the phase-locking phenomenon, the finite genus of the spectral curve corresponds to the finite number of Arnold tongues.

Another class of the flows corresponds to the isomonodromic evolution when the monodromy of the solutions to the Heun or Hill equations are fixed. In the RSJ model example, isomonodromic flows are discussed in~\cite{bibilo2022families} and are related to the solutions of the Painleve equations.

Remark that there exists a one-to-one correspondence between the Virasoro coadjoint orbits and the moduli space of two-dimensional hyperbolic metrics, see, for example~\cite{blau2024} for a detailed discussion. The orbits with elliptic monodromy correspond to the conical geometries, with hyperbolic monodromies to the annular geometries, and orbits with parabolic monodromies to the cuspidal geometries. The special orbits without constant representative correspond to the new topological sectors of two-dimensional gravity, characterized by twisted boundary conditions. In summary, the flows between the coadjoint orbits discussed above correspond to the flows between the different 2D geometries with possible modification of the boundary conditions. Certainly a detailed analysis of the flows in the $(B,A)$ plane in RSJ modes is required.

\section{Rotation number quantization and  exact WKB}\label{sec:QM-Rho}

In this Section we first argue that the rotation number quantization
is nothing but the exact WKB quantization in QM which can be easily seen with the help of so-called Milne anzatz. Remark that there are several approaches for exact quantization scheme. Let us mention three of them; the exact WKB \cite{voros1983return,jentschura2004instantons,zinn2004multi}, the uniform WKB \cite{langer1934asymptotic,alvarez2004langer,dunne2014uniform} and the Milne quantization \cite{milne1930numerical,korsch1985milne}. The anzatz for the wave function in these approaches is different. The exact WKB has a more transparent link with the classical spectral curve, while the link with the rotation number quantization is more transparent in the Milne quantization approach. We start with a reminder about the exact WKB approach to quantization and then discuss its relation to the Milne ansatz. We focus on the cases where the potential is periodic, so instead of energy-level quantization, we have the band-gap structure. In the case of non-periodic potential, all the relations between Milne quantization approach and exact WKB method hold, too.

\subsection{Exact WKB}

Consider the general Schr\"{o}dinger equation
\begin{equation}\label{eq:Schrod-general}
    -\hbar^2\frac{d^2\psi}{dx^2} + \left(V(x)-E\right)\psi(x)=0,
\end{equation}
where $\hbar$ is dimensionless Planck constant, $V(x)$ and $E$ are dimensionless (i.e. rescaled by energy unit) potential and energy, respectively. To analyze this equation, one can introduce the exact WKB ansatz \cite{grassi2020non}
\begin{equation}\label{eq:WKB-wavefunction}
    \psi(x)=\exp\left\{\frac{i}{\hbar}\int^xdx'\,Q(x')\right\}.
\end{equation}
Substituting the ansatz~\eqref{eq:WKB-wavefunction} into the eq.~\eqref{eq:Schrod-general}, we see that the function $Q(x)$ obeys the Riccati equation.
\begin{equation}\label{eq:WKB-Riccati}
    Q^2 -i\hbar\frac{dQ}{dx}=K(x)^2,\quad K(x) = \sqrt{E-V(x)}.
\end{equation}
In the exact WKB approach, the quantity $Q(x)$ is complex and should be expanded to the formal power series with respect to $\hbar$. We can decompose $Q(x)$ as the sum of two contributions, odd and even powers of $\hbar$, $Q(x)=Q_{\text{odd}}(x)+Q_{\text{even}}(x)$. Substituting this decomposition into the Riccati equation~\eqref{eq:WKB-Riccati}, we obtain system of two equations,
\begin{equation}\label{eq:WKB-Q-system}
\begin{gathered}
    Q_{\text{even}}(x)^2+Q_{\text{odd}}^2-i\hbar\frac{dQ_{\text{odd}}(x)}{dx}=K(x)^2,\\
    2Q_{\text{even}}(x)Q_{\text{odd}}(x)-i\hbar\frac{dQ_{\text{even}}(x)}{dx}=0.
\end{gathered}
\end{equation}
The second line tells us that the odd contribution $Q_{\text{odd}}(x)$ is the total derivative of the logarithm of the even contribution up to numerical factor $i\hbar/2$. If we set $Q_{\text{odd}}(x)=-i\hbar(\ln f(x))'$, we can represent the even contribution as $Q_{\text{even}}(x)=Cf(x)^{-2}$, where $C$ is the integration constant, and we can set $C=1$ without loss of generality. Using the mentioned decomposition, denoting $Q_{\text{even}}(x)\equiv P(x)$, and taking into account that $Q_{\text{odd}}$ is the total derivative, we can rewrite the WKB ansatz~\eqref{eq:WKB-wavefunction} as
\begin{equation}\label{eq:P-WKB-wavefunction}
    \psi(x) = \frac{1}{\sqrt{P(x)}}\exp\left\{\frac{i}{\hbar}\int^x dx'\,P(x')\right\}.
\end{equation}
The quantity $P(x)$ is called \emph{quantum-corrected momentum} and should be treated as the power series $P(x)=\sum_np_n(x)\hbar^{2n}$. The quantity $p_0(x)$ coincides with the classical momentum. From the geometric point of view, the quantity $dx\,P(x)$ is the meromorphic differential on the elliptic curve $(x,y)$ defined as
\begin{equation}
    y^2=E-V(x).
\end{equation}
Note that this curve also defines the isoenergetic surface. Therefore, in the exact WKB anzatz the classical differential $dx\,p_0(x)$ is replaced by the quantum differential $dx\,P(x)$. Note that strictly speaking the expression for the wave function~\eqref{eq:P-WKB-wavefunction} should contain two constants, because it is a \emph{general} solution of the second order ODE. The expansion of the quantum momentum in terms of the classical spectral curve can be found, for instance, in \cite{mironov2010nekrasov}. We note that exact WKB quantization was introduced long ago \cite{voros1983return,jentschura2004instantons,zinn2004multi} and has attracted attention recently due to new approaches to the resurgence of series in $\hbar$ \cite{grassi2016topological,grassi2020non,dunne2017wkb}. The
exact WKB is closely related to the TBA equations \cite{ito2025exact}, the
cluster algebras \cite{iwaki2014exact}, the Painleve equations \cite{del2023threefold} and the generic spectral networks \cite{gaiotto2013spectral}. This is not a complete list of applications and references.

\subsection{Milne quantization approach}

In the paper \cite{milne1930numerical} by Milne, the author observed that the energy spectrum in the Schr\"{o}dinger equation~\eqref{eq:Schrod-general} can be obtained from the solution of the EP equation. Despite the fact that periodicity of the potential $V(x)$ is not required, we focus on the periodic potential $V(x)$, so the Schr\"{o}dinger equation is the Hill equation. The corresponding EP equation is
\begin{equation}\label{eq:EP-Milne-form}
    w''(x)+k(x)^2w(x)=w(x)^{-3},\quad k(x)=\frac{\sqrt{E-V(x)}}{\hbar}.
\end{equation}
The general complex solution $\psi(x)$ of the Hill equation can be written in the form
\begin{equation}\label{eq:General-WKB}
    \psi(x) = \mathcal{A}w(x)\exp\left\{i\int^x\frac{dx'}{w(x')^2}+i\alpha\right\},
\end{equation}
where $\mathcal{A}$ and $\alpha$ are constants. Comparison of Milne's anzatz~\eqref{eq:General-WKB} and the corresponding EP equation~\eqref{eq:EP-Milne-form} to the exact WKB anzatz~\eqref{eq:P-WKB-wavefunction}, gives us the relation between $P(x)$ and $w(x)$, $P(x)=\hbar w(x)^{-2}$. The additional power of $\hbar$ comes from the definition of quantity $k(x)$ and should not be confused.

Since we know that the solution of the EP equation can be expressed via two linearly independent solutions of the Hill equation (see eq.~\eqref{eq:EP-general-solution}), the appeared integral in~\eqref{eq:General-WKB} can be evaluated explicitly \cite{korsch1985milne},
\begin{equation}\label{eq:Explicit-Milne}
    \int_{x_1}^{x_2}\frac{dx}{w(x)^2} = -\left.\arctan\left[W\left(b+a\frac{u_1(x)}{u_2(x)}\right)\right]\right|_{x_1}^{x_2},
\end{equation}
where $W$ is the Wronskian of two solutions $\psi_1(x)$ and $\psi_2(x)$, $a$ and $b$ are two constant that depend on the initial conditions. According to \cite{johnson1982rotation}, we can define the rotation number for the respective dynamical system on torus by using the eq.~\eqref{eq:General-WKB} as
\begin{equation}
    \rho = 2\lim\limits_{x\rightarrow\infty}\frac{\arg\psi(x)}{x}=\lim\limits_{x\rightarrow\infty}\frac{2}{x}\int_{x_0}^x\frac{dx'}{w(x')^2}=\lim\limits_{x\rightarrow\infty}\frac{2}{x}\left.\arctan\left[W\left(b+a\frac{\psi_1(x)}{\psi_2(x)}\right)\right]\right|_{x_0}^{x},
\end{equation}
where $\psi_1(x)$ and $\psi_2(x)$ are two real independent solutions. We know that the rotation number does not depend on the initial value $x_0$, so we can set $x_0=0$. Then, the inverse tangent is determined up to $\pi k$, $k\in\mathbb{Z}$, we need to choose a correct branch of this function. This choice implies the counting of poles on the whole real line. The number of poles is nothing more than the number of zeros in the function $\psi_2(x)$. So, we again arrive at the relation between the rotation number and the number of zeros of the Hill equation solution. It is not a surprise because due to the Sturm comparison theorem the density of states (i.e. the number of eigenvalues that are less than a given value $E$ normalized by the interval length) is directly related to the number of solution zeros (they coincide up to addition of $\pm 1$), so the rotation number coincides with the density of states up to the addition of unity, which was noticed in \cite{johnson1982rotation} and stated earlier. In the case of periodic potential, Milne's quantization corresponds to the quantization of bands. Each band is labeled by the number of zeros of the Hill equation solution in this band. Again, we finalize this section by the comparison table; see Tab.~\ref{tab:Band-Gap-Rho}.

\begin{table}[h!]
\centering
\begin{tabular}{|c|c|c|c|}
\hline
zone structure & rotation number        & Hill monodromy  \\ \hline
band           & $\rho \in \mathbb{Z}_{>0} $             & hyperbolic \\ \hline
gap            & $\rho\notin\mathbb{Z}$ & elliptic   \\ \hline
band-gap edge  & ---                    & parabolic  \\ \hline
\end{tabular}
\caption{Correspondence between band-gap structure, rotation number and Hill equation monodromies.}
\label{tab:Band-Gap-Rho}
\end{table}

\section{Slow-fast dynamics and semiclassical WKB}\label{sec:Slow-Fast-WKB}

In this section, we discuss how the slow-fast behavior in classical dynamical systems can be mapped to the WKB analysis in quantum systems. The small parameter in the dynamical system plays the role of the Planck constant in the corresponding Schr\"{o}dinger equation. This approach looks quite prominent because it allows one to apply some well-known results and ideas from quantum mechanics and obtain new insights related to the slow-fast effects. We focus on dynamical systems on torus of the M\"{o}bius type and use the Mathieu equation and the Hill equation for the RSJ model as working examples.

\subsection{Slow-fast dynamics and canard phenomenon}

In the dynamical system theory a general slow-fast system can be represented as
\begin{equation}
    \dot{x}= F(x,y), \quad \dot{y}=\epsilon G(x,y)
\end{equation}
where $\epsilon\ll 1$. The variable $x$ is called a fast variable, while the variable $y$ is a slow variable. In the leading approximation $y=\mathrm{const}$ and the so-called slow manifold $F(x,y)=0$ appears. The dynamics of a slow-fast system is as follows. The flow is concentrated mainly near the slow manifold when $y$ is considered as the fixed parameter. The slow manifold involves stable and unstable components, and slow-fast behavior favors the existence of the so-called \emph{canards} \cite{diener1984canard,callot1981chasse}. They correspond to the jumps in the slow manifold at some fixed values of $y$. 

One of the simplest possible examples for the slow-fast dynamics with canards is forced Van der Pol (VdP) oscillator \cite{desroches2011canards}, which has the following equation of motion,
\begin{equation*}\label{eq:VdP-slow-fast}
    \ddot{x}+\mu\left(1-x^2\right)\dot{x}+x=a,
\end{equation*}
where $\mu\gg 1$, $a=\mathrm{const}$. Denoting $\epsilon=\mu^{-1}$, introducing the variable $y=\epsilon\dot{x}+x^3/3-x$ and rescaling the time as $t\rightarrow t/\epsilon$, this equation becomes
\begin{equation}
    \dot{x}=y-\frac{x^3}{3}+x,\quad \dot{y}=\epsilon\left(a-x\right).
\end{equation}
The slow manifold is given by the curve $S(x,y)$,
\begin{equation*}
    S(x,y) = \left\{(x,y): y-\frac{x^3}{3}-x=0\right\}.
\end{equation*}
The curve $S$ has two attractive branches and one repellent branch. They are separated by fold points that correspond to saddle-node bifurcation points of the fast subsystem. The dynamics is slow near the curve $S$, but in general can be supplemented by extended fragments at fixed $y$ between the attractive and repellent parts of the curve $S$. Depending on the parameters $(\epsilon,a)$ there can be a limit cycle in the form of relaxation oscillation, a canard with head and a canard without head. The case of a dynamical system on the torus is more interesting. In the paper \cite{guckenheimer2001duck} the authors have shown that an auxiliary parameter (such as an external force in the case of the VdP oscillator) is not necessary for the existence of canards. This research was continued in \cite{shchurov2010canard,schurov2017duck} and in \cite{kleptsyn2013josephson}, where the first attempt was made to analyze slow-fast dynamics in the RSJ model.

The main feature of canards is that they are effects of order $\exp(-1/\epsilon)$, which implies the non-perturbative nature of the phenomenon. Using the correspondence between dynamical system on torus of M\"{o}bius type and Hill equation, we will argue that the complexified slow manifold coincides (upon change of variables) with the spectral curve, defined by the Schr\"{o}dinger equation with a given Hill potential. It means that the canards in dynamical systems presumably correspond to the non-perturbative effects in the quantum mechanics: instantons and anti-instantons. We will consider the different aspects of the interplay between canards and instanton-like effects elsewhere.

\subsection{Slow-fast dynamics versus Schr\"{o}dinger equation}

Here we provide the correspondence between the slow-fast dynamics in the dynamical system on torus of M\"{o}bius type and the semiclassical limit of the Schr\"{o}dinger equation. We will parallel the representation of the solution to the slow-fast dynamical system and the exact WKB anzatz in quantum mechanics. The small parameter in the dynamical system is mapped with the Planck constant. Using the intuition from the semiclassical approximation in quantum mechanics, we will argue that the complexification of the slow manifold which is the counterpart of the spectral curve of the complex Hamiltonian system provides the proper framework for the analysis of the widths of the phase-locking domains as well as the canard-type phenomena. We start with the Hill equation,
\begin{equation}
    \frac{d^2u}{dt^2}+V(t)u=Eu(t),
\end{equation}
where $V(t)$ is the periodic function with period $2\pi/\omega$ and $E$ is the parameter. Denoting $\tau=\omega t$ and rescaling all the quantities we obtain
\begin{equation}\label{eq:Hill-SlowFast}
    \omega^2\frac{d^2u}{d\tau^2}+V(\tau)u(\tau)=Eu(\tau),
\end{equation}
where all the coefficients are dimensionless. The equation~\eqref{eq:Hill-SlowFast} can be represented as the system of first order equations,
\begin{equation}\label{eq:Hill-Phi-omega}
    \dot{u} = v,\quad \dot{v} = \left(E-V(\tau)\right)u,\quad \dot{\tau}=\omega.
\end{equation}
In the limit $\omega\rightarrow 0$, the variable $\tau$ is \emph{slow}, while the variables $u$ and $v$ are \emph{fast}. In order to represent the slow manifold in a convenient form, we introduce the quantity $\Phi=-v/u$, that brings the system~\eqref{eq:Hill-Phi-omega} into the Riccati equation,
\begin{equation}\label{eq:Hill-Phi-tau}
    \dot{\Phi} = -E+\Phi^2+V(\tau),\quad \dot{\tau}=\omega.
\end{equation}
The quantity $\Phi(\tau)$ can be also introduced by using the ansatz
\begin{equation}\label{eq:Slow-fast-Riccati}
    u(\tau)=\exp\left\{-\frac{1}{\omega}\int d\tau\,\Phi(\tau)\right\},
\end{equation}
that transforms~\eqref{eq:Hill-SlowFast} into the Riccati equation for variable $\Phi(\tau)$, eq.~\eqref{eq:Hill-Phi-tau}. The slow manifold is determined by the condition $\dot{\Phi}=0$, so it is defined by the relation
\begin{equation}\label{eq:Hill-E-surface}
    E=\Phi^2+V(t).
\end{equation}
This relation~\eqref{eq:Hill-E-surface} shows that the slow manifold determined by $\dot{\Phi}=0$ in eq.~\eqref{eq:Hill-Phi-tau} is nothing more than the fixed energy condition for the single non-relativistic degree of freedom in potential $V(t)$. The equation for $\Phi(t)$ in~\eqref{eq:Hill-Phi-tau} can be matched with the dynamics on the torus. Indeed, let $\phi(t)=2\arctan\Phi(t)$, then the phase variable $\phi(t)$ obeys the equation
\begin{equation}
    \dot{\phi} = \left(V(t)-E+1\right) + \left(V(t)-E-1\right)\cos\phi,
\end{equation}
which we have obtained earlier (cf. with eq.~\eqref{eq:Hill-DynTorus}). 

As we have noted above, the classical spectral curve which for the one-dimensional system is the fixed energy manifold in the complexified phase space is the central object in the exact quantization and semiclassical limit. To obtain the explicit results in exact WKB, we integrate the classical or quantum momentum over the cycles on this Riemann surface. In the simplest example of the elliptic spectral curve which occurs in the Mathieu example, there are two types of cycles, $\alpha$ and $\beta$ cycles. Generally speaking, the first type contains real turning points, whereas the second type contains complex turning points. The integration over the ``perturbative'' $\alpha$-cycle gives the exact WKB quantization, while the integration over ``non-perturbative'' $\beta$-cycle gives the instanton-like contributions, for instance the narrow gaps at high energies. For a generic situation, the genus of the Riemann surface can be high or even infinite. To match the slow-fast dynamics approach with the WKB approach, we treat $\Phi$ and $t$ as complex variables. The ideas described above are related to the discussion in \cite{kristiansen2024dynamical,kristiansen2025dynamical}, where the authors start with the theory of dynamical systems and then apply ideas from quantum mechanics. In contrast, we start from the quantum mechanical problem that can be matched with the dynamical system.

It is worth commenting on the relation between the complexified slow manifold and the naive torus pictured in Fig.~\ref{fig:dynamics-on-torus}. The naive torus involves $\phi\in S^1$, $\tau\in S^1$ and does not have well-defined moduli of the complex structure. In contrast, the complexified slow manifold $\Phi\in\mathbb{C}$, $t\in\mathbb{C}$ supplemented by the single complex equation~\eqref{eq:Hill-Phi-tau} is a Riemann surface whose genus depends on a particular model under consideration. The moduli of the complex structure of this surface depend on the parameters of the dynamical system, and generically there are some points on the moduli space when the surface degenerates. The Riemann surface is defined in the slow-fast limit, but similarly to exact WKB defines the behavior beyond this limit, only the variable $\Phi$ that obeys the Riccati equation takes a more complicated form.

Assuming $\Phi\in\mathbb{C}$ in~\eqref{eq:Hill-Phi-tau}, let us again write $\Phi(\tau)=\Phi_{\text{even}}(\tau)+\Phi_{\text{odd}}(\tau)$, where the even contribution contains only even powers of $\omega$, whereas the odd contains only odd powers. Then, the Riccati equation becomes system of two equations for $\Phi_{\text{even}}(\tau)$ and $\Phi_{\text{odd}}(\tau)$,
\begin{equation}
    \omega\frac{d\Phi_{\text{odd}}}{d\tau}=-E+\Phi_{\text{even}}^2+\Phi_{\text{odd}}^2+V(\tau),\quad   \omega\frac{d\Phi_{\text{even}}}{d\tau}=2\Phi_{\text{even}}\Phi_{\text{odd}}.
\end{equation}
In this equation let us denote $K(\tau)^2=E-V(\tau)$, then we have
\begin{equation}
    -\omega\frac{d\Phi_{\text{odd}}}{d\tau}+\Phi_{\text{even}}^2+\Phi_{\text{odd}}^2=K(\tau)^2,\quad \omega\frac{d\Phi_{\text{even}}}{d\tau}=2\Phi_{\text{even}}\Phi_{\text{odd}}.
\end{equation}
The second equation says that $\Phi_{\text{odd}}$ is the total derivative of the logarithm of $\Phi_{\text{even}}$ multiplied by $\omega/2$. Setting $\Phi_{\text{odd}}(\tau)=-\omega(\ln f(\tau))'$ (cf. with quantum mechanics), we find $\Phi_{\text{even}}(\tau)=Cf(\tau)^{-2}$, where $C$ is the integration constant, and again we set $C=1$. Denoting $\Phi_{\text{even}}(\tau)\equiv F(\tau)$, we obtain the following expression for the solution $u(\tau)$,
\begin{equation}
    u(\tau) = \frac{1}{\sqrt{F(\tau)}}\exp\left\{-\frac{1}{\omega}\int^{\tau}d\tau'F(\tau')\right\}.
\end{equation}
Also, we know that the function $f(\tau)$ obeys the EP equation,
\begin{equation}
    -\omega^2\frac{d^2f}{d\tau^2}+K(\tau)^2f(\tau)=f(\tau)^{-3}.
\end{equation}
When $\omega\rightarrow 0$, we obtain the relation $F(\tau)=f(\tau)^{-2}=\sqrt{E-V(\tau)}$, as should be. Finally, we can represent the ansatz for $u(\tau)$ in terms of $f(\tau)$,
\begin{equation}
    u(\tau)=f(\tau)\exp\left\{-\frac{1}{\omega}\int^{\tau'}\frac{d\tau}{f(\tau)^2}\right\}
\end{equation}
The summary of the relation between slow-fast dynamics and exact WKB ansatz is given in Tab.~\ref{tab:slow-fast-WKB}.
\renewcommand{\arraystretch}{2.0}
\begin{table}[h!]
\hspace*{-2em}
\begin{tabular}{|c|c|c|}
\hline
& Quantum mechanics & Slow-fast system \\[1.4ex] \hline
equation & $\displaystyle-\hbar^2\psi''+V(x)\psi(x)=E\psi(x)$ & $\displaystyle\omega^2\ddot{u}+V(\tau)u=Eu$ \\[1.2ex] \hline
ansatz & $\displaystyle\psi(x) = \exp\left\{\frac{i}{\hbar}\int^xdx'\,Q(x')\right\}$\ & $\displaystyle u(\tau)=\exp\left\{-\frac{1}{\omega}\int^{\tau} d\tau'\,\Phi(\tau')\right\}$ \\[1.4ex] \hline
decomposition   & $Q(x)=Q_{\text{odd}}(x)+P(x)$ & $\Phi(\tau) = \Phi_{\text{odd}}(\tau)+F(\tau)$ \\[1.4ex] \hline
reduced ansatz & $\displaystyle\psi(x)=\frac{1}{\sqrt{P(x)}}\exp\left\{\frac{i}{\hbar}\int^x dx'\,P(x)\right\}$ & $\displaystyle u(\tau)=\frac{1}{\sqrt{F(\tau)}}\exp\left\{-\frac{1}{\omega}\int^\tau d\tau'\,F(\tau')\right\}$ \\[1.4ex] \hline
EP ansatz  & $\displaystyle\psi(x)=f(x)\exp\left\{\frac{i}{\hbar}\int^x\frac{dx'}{f(x')^2}\right\}$ & $\displaystyle u(\tau)=f(\tau)\exp\left\{-\frac{1}{\omega}\int^{\tau}\frac{d\tau'}{f(\tau')^2}\right\}$ \\[1.4ex] \hline
EP equation & $\hbar^2f''+K(x)^2f=f^{-3}$ & $-\omega^2\ddot{f}+K(\tau)^2f=f^{-3}$ \\ \hline
\end{tabular}
\caption{Relations between exact WKB ansatz in quantum mechanics and slow-fast dynamics for system on torus}
\label{tab:slow-fast-WKB}
\end{table}
\renewcommand{\arraystretch}{1.0}

The link between the QM with periodic potential and the dynamical system on the torus suggests the possibility of using results concerning the exact and semiclassical WKB for the dynamical system. The key point is the complexification of the ``phase space'' variables in the dynamical system $(\Phi,t)\in\mathbb{C}$, therefore the slow manifold defined by the condition $\dot{\Phi}=0$ is the Riemann surface of the particular genus which depends on the potential in the Hill equation. The complexification of phase space is quite familiar in the non-perturbative analysis of the Hamiltonian systems. On the other hand, looking from the perspective of the parametric oscillator, the complexification of the time variable has been reviewed in \cite{popov2005imaginary}. In particular, complex time instantons are relevant for non-perturbative effects on particle dynamics in the time-dependent electric field \cite{brezin1970pair,popov1972pair}.

The ``$p\,dx$'' differential now reads as the $\Phi\,d\tau$ differential on the Riemann surface, and we can apply some results familiar from the holomorphic Hamiltonian systems. First, assuming the simplest case of the torus spectral curve we introduce two action integrals over ``perturbative'' and  ``non-perturbative'' $\alpha$ and $\beta$ cycles
\begin{equation}
    a=\oint_{\alpha}d\tau\,\Phi \qquad  a_D=\oint_{\beta}d\tau\,\Phi
\end{equation}
which are useful to describe the band structure quantitatively.

That is, the exact equations for the boundaries of the bands in QM are analogues of the equations for the boundaries of the phase-locking domains in the dynamical system. The key element which provides the exact analytic expressions for the band boundaries is the auxiliary Riemann surface equipped with the particular meromorphic differential. This Riemann surface known as the spectral curve gets identified with the energy level and the corresponding differential is the ``$p\,dx$'' one. In QM the classical momentum $p_0(x)$ is replaced by $P(x)$, while the Riemann surface is not modified.

\subsection{Mathieu and RSJ examples}

\subsubsection{Mathieu equation}

We have the following Mathieu equation,
\begin{equation}\label{eq:Mathieu-dimensionful}
    \frac{d^2u}{dt^2}-\left(B-A\cos\omega t\right)u=0.
\end{equation}
This system contains three parameters: the constant shift $B$, the amplitude of potential $A$, and the potential frequency $\omega$. We have already shown the band-gap structure, see Fig.~\ref{fig:Mathieu-Plots}, but there exists an interesting behavior in the domains $B/A>-1$ and $B/A<-1$, see Fig.~\ref{fig:Mathieu-omegas}. In the domain with $B/A<-1$ the phase-locking domains are very thin, whereas at $B/A>-1$ the gaps between phase-locking domains are very thing.
\begin{figure}
    \centering
    \includegraphics[width=\linewidth]{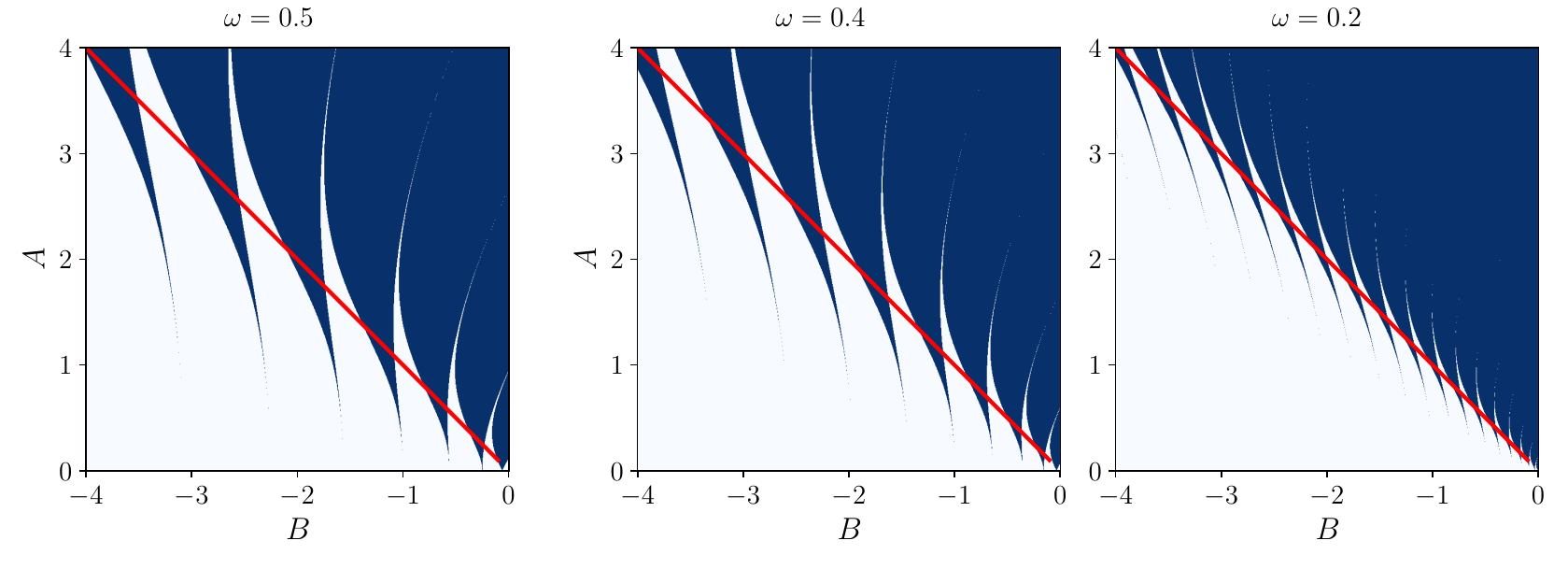}
    \caption{Phase-locking domains (shaded) for the Mathieu equation with small frequencies $\omega$. The line corresponds to $A=-B$ and separate two domains}
    \label{fig:Mathieu-omegas}
\end{figure}
In order to provide the connection between the slow-fast dynamics and the exact WKB ansatz, we notice that, in fact, the number of parameters is two, instead of three: we can use the ratio $B/A$ and $\omega$. It is convenient to denote $E=B/A$ and rescale $\omega\rightarrow\omega/\sqrt{A}$. Using these denotations, we rewrite the eq.~\eqref{eq:Mathieu-dimensionful} as
\begin{equation}\label{eq:Mathieu-dimensionless}
    \omega^2\frac{d^2u}{d\tau^2}-\left(E-\cos\tau\right)u=0,
\end{equation}
which can be rewritten as system of three first order equations,
\begin{equation*}
    \dot{u}=v,\quad \dot{v}=E-\cos\tau,\quad \dot{\tau}=\omega.
\end{equation*}
We plot the band-gap structure for this equation in Fig.~\ref{fig:Mathieu-WKB-Contours}. Notice that in terms of variable $E$, we have very thin Arnold tongues at $E<-1$ and very thin spacings between tongues at $E>-1$ for small $\omega$, which is consistent with Fig.~\ref{fig:Mathieu-omegas}. Then, using the projective coordinate, $\Phi=-v/u$, we obtain the equation for $\Phi$,
\begin{equation}\label{eq:Mathieu-slow manifold}
    \dot{\Phi}=-E+\cos\tau+\Phi^2,
\end{equation}
and slow manifold is defined by condition $\dot{\Phi}=0$. Then, we introduce the quantity $y=\left(e^{i\tau}-e^{-i\tau}\right)/2$, which allows us represent the complexified slow manifold $E=\Phi^2+\cos\tau$ from~\eqref{eq:Mathieu-slow manifold} as the elliptic curve, determined by the equation
\begin{equation}\label{eq:Mathieu-curve}
    y^2 = \left(E-\Phi^2\right)^2-1.
\end{equation}
This elliptic curve defines the 2D torus with two independent cycles, $\alpha$ and $\beta$. We consider the meromorphic one-form $\Omega=\Phi(\tau)\,d\tau=i\sqrt{E-\cos\tau}\,d\tau$, and   the integrals over the cycles $\alpha$ and $\beta$ are called \emph{action} and \emph{dual action} variables respectively,
\begin{equation}
    a=\oint_{\alpha}\Omega,\quad a_D=\oint_{\beta}\Omega.
\end{equation}
The cycles $\alpha$ and $\beta$ correspond to the contours in the complex plane that encircle the turning points of the potential $V(\tau)$. By definition, \emph{real} turning points $\tau_R$ are points where $\Phi(\tau)=0$, so they are given by
\begin{equation*}
    E-\cos\tau=0 \rightarrow \tau=\pm\tau_R=\pm\arccos E.
\end{equation*}
In addition to the real turning points, there are two \emph{imaginary} (complex) turning points that are $\tau=\pm\tau_{I}=\pm\pi$.
\begin{figure}
\begin{minipage}{0.49\linewidth}
    \centering
    \includegraphics{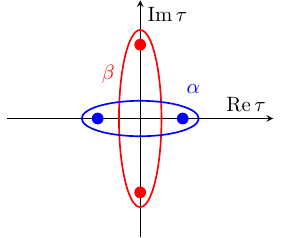}
\end{minipage}
\begin{minipage}{0.49\linewidth}
    \centering
    \includegraphics[width=0.8\linewidth]{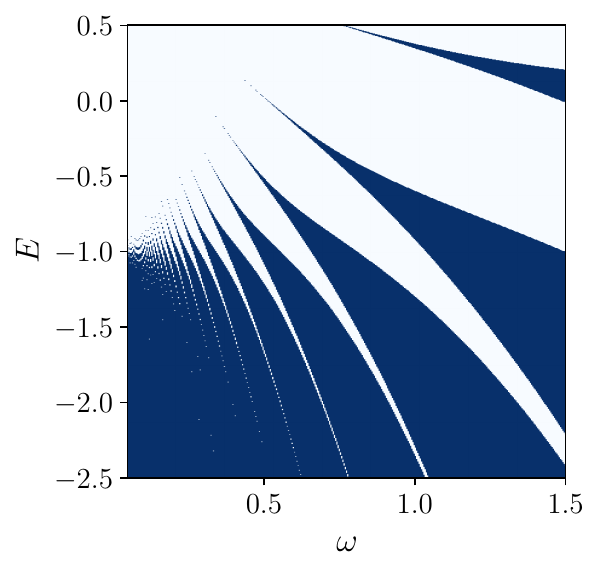}
\end{minipage}
    \caption{Left: two integration contours for the Mathieu equation that encircling real and complex turning points; right: band-gap structure of dimensionless Mathieu equation~\eqref{eq:Mathieu-dimensionless} (bands are shaded, gaps are empty)}
    \label{fig:Mathieu-WKB-Contours}
\end{figure}
The explicit form of these integrals is given by
\begin{equation*}
    a = 2i\int_{0}^{\tau_R}d\tau\,\sqrt{E-\cos\tau},\quad a_D = 2i\int_{0}^{\tau_I}d\tau\,\sqrt{E-\cos\tau}.
\end{equation*}
and they can be expressed in terms of the elliptic functions
\begin{equation}\label{eq:cycle-int-zero-order}
\begin{gathered}
    a(E)=\frac{4i}{\pi}\left[\bm{E}\left(\frac{1+E}{2}\right)-\frac{1-E}{2}\bm{K}\left(\frac{1+E}{2}\right)\right],\\
    a_D(E)=\frac{4}{\pi}\left[\bm{E}\left(\frac{1-E}{2}\right)-\frac{1+E}{2}\bm{K}\left(\frac{1-E}{2}\right)\right],
\end{gathered}
\end{equation}
where $\bm{K}(z)$ and $\bm{E}(z)$ denotes the complete elliptic integrals with parameter $z$ of the first and second kind, respectively. Note that in fact these expressions represent the leading-order terms, and they have to be denoted as $a^0(E)$ and $a_{D}^0(E)$. We omit the upper index for brevity.

The exact WKB approach tells us that during ``quantization'' (i.e., when we take into account all orders of $\omega$ in perturbative expansions) the elliptic curve~\eqref{eq:Mathieu-curve} remains the same and we only need to replace the ``classical'' variable $\Phi(\tau)$ by its ``quantum'' analog. To do it, we first make it quantum, setting $\Phi(\tau)=\Phi_{\text{odd}}(\tau)+F(\tau)$ and use the established above relation $F(\tau)=w(\tau)^{-2}$, where $w(\tau)$ is the solution of the EP equation. It gives the following expressions for integrals over $\alpha$- and $\beta$-cycles,
\begin{equation}\label{eq:Mathieu-integrals}
    a(E,\omega) = \oint_{\alpha}d\tau\,\Phi(\tau)=\int_{0}^{2\tau_R}\frac{d\tau}{w(\tau)^2},\quad a_D(E,\omega) = \oint_{\beta}d\tau\,\Phi(\tau)=\int_{0}^{2\pi}\frac{d\tau}{w(\tau)^2}.
\end{equation}
Let us first focus on the integral over the $\beta$-cycle in~\eqref{eq:Mathieu-integrals}. Since $w(\tau)$ is the solution of EP equation, we can immediately apply the Theorem~\ref{th:Hill-invariants}(see~\autoref{app:Inv-Virasoro}) and write down the following relation,
\begin{equation}
    \lambda_{1,2} = \exp\left\{\pm a_D(E,\omega)\right\},
\end{equation}
where $\lambda_{1,2}$ are eigenvalues of the monodromy matrix of the eq.~\eqref{eq:Mathieu-dimensionless}. Using the relation between the Floquet exponents and the eigenvalues, we obtain
\begin{equation}
    \nu_{1,2}=\frac{1}{2\pi}\ln\lambda_{1,2}=\pm\frac{a_D(E,\omega)}{2\pi},
\end{equation}
and this observation is consistent with the results in~\cite{he2012mathieu,he2015combinatorial}, obtained for Floquet exponents in perturbative manner. This means that the rotation number in the case of elliptic monodromy can be represented as $\rho = a_D(E,\omega)/(\pi i)$ (cf. with eq.~\eqref{eq:rotation-number-Delta}). Now, let us focus on the integral over the $\alpha$-cycle. Using the proposed above correspondence between slow-fast dynamics and the WKB approach to the Schr\"{o}dinger equation, we can use results related to the exact WKB for the Mathieu equation~\cite{dunne2017wkb,basar2015resurgence}. First of all, the all-orders WKB formulas, derived in~\cite{dunham1932wentzel}, defines the conditions for gaps and bands,
\begin{equation}\label{eq:Mathieu-band-gap-loc}
\begin{gathered}
    a(E,\omega)  = \oint_{\alpha}d\tau\,\Phi(\tau) = i\omega N\quad \text{for gaps}\\
    a(E,\omega) = \oint_{\alpha}d\tau\,\Phi(\tau)  = i\omega \left(N+\frac{1}{2}\right)\quad \text{for bands}.
\end{gathered}
\end{equation}
These expressions can be easily understood with the help of Milne's ansatz, which was implicitly used in~\cite{sukhatme1999}. Complete quantum analysis of the spectral band structure is available \cite{zinn2004multi, basar2015resurgence,gorsky2018bands}.

Adopting results for the exact WKB for the Mathieu equation, we can also estimate the width $\Delta$ of phase-locking domains (bands). To do it, we use the exact conditions~\eqref{eq:Mathieu-band-gap-loc} and the expression for $a$ and $a_D$ in the lowest order in $\omega$, see eq.~\eqref{eq:cycle-int-zero-order}. Using the results from~\cite[section 5]{dunne2017wkb}, we obtain the following expression for the width of phase-locking domains at $E\approx -1$,
\begin{equation}
    \Delta\sim  \sqrt{\frac{2}{\pi}}\frac{2^{4N+1}}{N!}\left(\frac{2}{\omega}\right)^{N-1/2}\exp\left(-\frac{8}{\omega}\right).
\end{equation}
Now, we can return the original variables. First, $E=B/A$, i.e. we can treat the condition $E\approx -1$ as $A\approx -B$. Second, we return the original frequency, which gives
\begin{equation}
    \Delta\sim\sqrt{\frac{2}{\pi}}\frac{2^{4N+1}}{N!}\left(\frac{2\sqrt{A}}{\omega}\right)^{N-1/2}\exp\left(-\frac{8\sqrt{A}}{\omega}\right).
\end{equation}
Next, we can estimate the width $\delta$ of the gaps between the phase-locking domains (spacings between Arnold tongues) in the limit of $E\gg 1$ or $N\omega\gg 1$. Again, we use the exact conditions for gaps and bands and expressions~\eqref{eq:cycle-int-zero-order}. This gives the following expression,
\begin{equation}
    \delta\sim \frac{\omega^2N}{2\pi A}\left(\frac{e\sqrt{A}}{\omega N}\right)^{2N}.
\end{equation}
The important fact is that the slow-fast dynamics makes the Arnold tongues exponentially thin in domain $B/A<-1$, which is consistent with numerical simulations. As we know that canards are associated with behavior $\exp(-1/\omega)$, we connect these exponentially thin tongues with canards, which is consistent with the discussion in~\cite{guckenheimer2001duck}.

\subsubsection{Slow-fast dynamics in RSJ}

The analysis of slow-fast dynamics in the RSJ model was started in the papers by \cite{guckenheimer2001duck,shchurov2010canard} and continued in \cite{klimenko2013asymptotic,kleptsyn2013josephson}. Let us summarize the structure of phase-locking domains in RSJ in the slow-fast regime on the $(A,B)$ parameter plane (described in \cite{kleptsyn2013josephson}),
\begin{table}[h!]
\begin{tabular}{|c|c|c|}
\hline
           & parameters  & properties of Arnold tongues                        \\ \hline
Domain I   & $A<B-1$     & very thin tongues                                    \\ \hline
Domain II  & $B-1<A<B+1$ & very thin spacings between tongues, no constrictions \\ \hline
Domain III & $A>B+1$     & parquet-like tongues with constrictions              \\ \hline
\end{tabular}
\end{table}

The behavior of Arnold tongues in these domains is shown in Fig.~\ref{fig:slow-fast-tongues}. We see that this behavior shares some similarity with the Mathieu equation (existence of very thin tongues and spacings between them). Moreover, the constrictions appear when $A>B+1$.
\begin{figure}[h!]
    \centering
    \includegraphics[width=\linewidth]{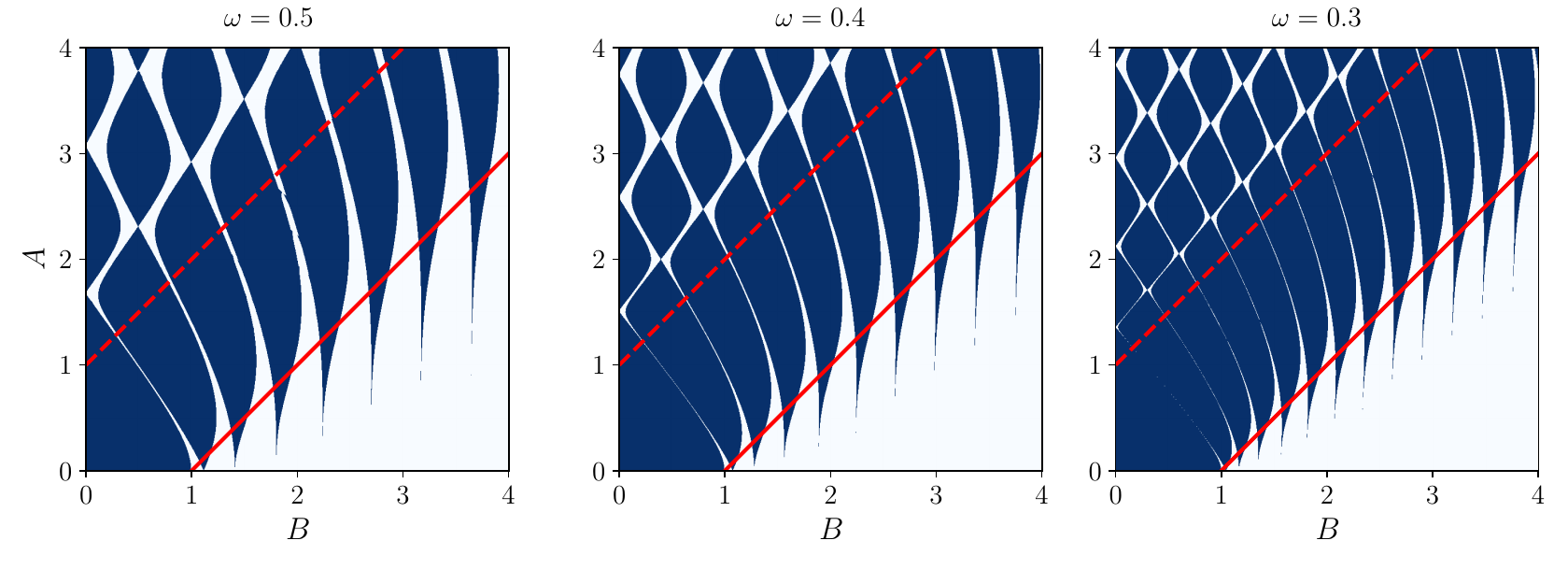}
    \caption{Arnold tongues for different values of $\omega$. Domain I corresponds to the region below the full red line; domain II corresponds to the region between full \& dashed red lines; domain III corresponds to the region above the dashed line (color online)}
    \label{fig:slow-fast-tongues}
\end{figure}
As was noticed in \cite{renne1974some}, in the domain I the potential $V(t)$ in Hill equation~\eqref{eq:RSJ-Hill-potential} in $\omega\rightarrow 0$ limit with $A\ll 1$ is well-approximated by the Mathieu potential,
\begin{equation}
    V(t)\approx \frac{B^2-1}{4}+\frac{AB}{2}\cos\omega t
\end{equation}
So, in this domain the results for the Mathieu equation holds. In general, we realize the following program:
\begin{enumerate}
\setlength{\itemsep}{0pt}
\setlength{\parskip}{0pt}
    \item Using the explicit form of the Hill equation potential $V(t)$ for the RSJ model, determine the spectral curve. Based on the known facts, this curve is still torus.
    \item Compute in the lowest order in $\omega$ the integrals over cycles.
    \item Using the exact WKB quantization results, derive the estimations for the width of Arnold tongues in the $A<B-1$ domain, the width of gaps between tongues in the $B-1<A<B+1$ domain.
    \item Keeping in mind that the potential has a singularity at $A>B+1$, relate the appearance of the constrictions and the parquet-like structure to the cycle integrals.
\end{enumerate}
In terms of the complexified slow manifold, the non-perturbative effect corresponds to the motion along the non-perturbative cycle instead of jump as we have discussed for generic slow-fast system. However, the Hill equation corresponding to the RSJ model has a quite complicated potential structure. The accurate and rigorous analysis will be presented elsewhere, but the idea is clear: the non-perturbative effects in $\omega$ in the RSJ model can be captured by the exact WKB approach and the complexification of the slow manifold.

\section{Discussion}\label{sec:Discussion}

Our study certainly suggests many directions for further research. Let us mention some of them.
\begin{itemize}
\setlength{\itemsep}{0pt}
\setlength{\parskip}{0pt}
\item Since we have linked the dynamical system of the torus to the quantum mechanics with periodic potential we can use the new tools and ideas developed during the last decades for the latter. As we have argued, the slow-fast dynamics can be effectively described in terms of the complexified slow manifold, which is analogue of the Seiberg-Witten curve in the context of the low-energy $\mathcal{N}=2$ SUSY YM theory \cite{seiberg1994electric}. It underlies the dynamics of the surface defect and naturally provides the quantum mechanical system in the Nekrasov-Shatashvili  limit of the $\Omega$-deformed theory \cite{nekrasov2003seiberg,nekrasov2010quantization}. 

The complexified slow manifold is bundled over the moduli space of complex structures. Generically, there are degeneration points on the moduli space where some cycle of the Riemann surface gets shrunk. The variables $(a_i,a_{D,i})$ form the symplectic pair in the so-called Whitham Hamiltonian system that governs the dynamics on the moduli space of the Riemann surface \cite{krichever1994tau,gorsky1998rg,edelstein1999whitham}. 
The Hamiltonian equation of motion in the Whitham system involves the so-called prepotential $\mathcal{F}$ playing the role of the action variable 
\begin{equation}
    a_{D,i}=\frac{\partial \mathcal{F} }{\partial a_i}, \qquad E= \frac{\partial \mathcal{F} }{\partial t}
\end{equation}
Here we denote the time variable in the Whitham equation of motion as $t$ however 
it depends on the problem under consideration. It would be interesting to develop the
Whitham approach for the generic dynamical system on the torus along these lines. 
It would use the relation with the topological string picture, since the prepotential in all such problems can be interpreted as the partition function of the topological string at some target manifold. The general relation between quantum mechanics and topological strings can be found
in \cite{grassi2016topological,codesido2019non}. 

As in the $\mathcal{N}=2$ SYM context, the action of a modular group
on the different objects related to the complexified slow manifold provides additional tools to analyze the behavior of the dynamical system.

\item If we consider the complexified slow manifold, there is a marginal stability curve which, for the simplest $\mathrm{SU}(2)$ example, reads as $\im (a_D/a)=0$ \cite{seiberg1994electric}
at the base of the fibration of the corresponding Riemann surface. It is the simplest example of generic wall crossing phenomena; see, for instance, \cite{kontsevich2014wall, gaiotto2010four}. Hence, we can expect a similar wall-crossing phenomenon for the dynamical system of the torus. The careful account of the Stokes lines and surfaces provides the proper way to approach these issues. In the RSJ model the Stokes phenomena have been discussed in \cite{bibilo2022families}. 

\item At the quantum level, there is a so-called P/NP relation that connects perturbative and non-perturbative contributions \cite{zinn2004multi,basar2015resurgence,basar2017quantum,ccavucsouglu2024resurgence} to different quantum-mechanical variables and provides an effective tool for the resurgence analysis in quantum mechanics. It was argued that the P/NP relation can be interpreted as the equation of motion in the Whitham theory \cite{gorsky2015rg}. The P/NP relation in the classical limit can be considered as the simple consequence of the PF equation. However, in the quantum case it has a bit more complicated geometrical interpretation; see, for instance,~\cite{dunne2017wkb}. For instance for the Mathieu case this reads as 
\begin{equation}
    \frac{\partial E(a,\hbar)}{\partial a} = \frac{i}{8\pi}\left(a_D(a,\hbar) - a\frac{\partial a_D(a,\hbar)}{\partial a}
    -\hbar \frac{\partial a_D(a,\hbar)}{\partial \hbar}\right)
\end{equation}
where $a(E,\hbar)$ and $a_D(E,\hbar) $ are the all-order WKB action and dual ``instantonic'' action. We expect similar P/NP relations for the expansion of dynamical systems in $\omega$, up to change $\hbar\rightarrow i\omega$.

\item The parameters in the dynamical system can be random variables. Using our mapping of the dynamical system to the Hill equation, we get the quantum mechanics with the random potential, which could develop the Anderson localization of the degrees of freedom. It would be interesting to discuss the possible relation of the Anderson localization with the phase-locking in a noisy dynamical system.

\item The dynamical systems describe non-trivial behavior of the RG flows and among many works in this field, we would like to mention \cite{gukov2017rg}, where the generic bifurcations in the RG flows are classified. The cycles in the RG are discussed in \cite{Leclair2003russian,Bulycheva2014spectrum,jepsen2021rg,motamarri2024refined,leclair2025non} and homoclinic orbits are found in \cite{jepsen2021homoclinic}. Finally, chaotic behavior was identified in \cite{bosschaert2022chaotic}. In our study we deal with dynamical systems on torus of M\"{o}bius type and here a Cantor staircase with integer steps appears. This is quite common in systems where the non-perturbative effects dominate in the RG flow.

It would be interesting to consider the ``non-perturbative effects'' in the RG flow corresponding to canard-like phenomenon in the dynamical system. To this end, the complexification of the RG time is required. In a holographic setting, the RG flow corresponds to the evolution along the radial coordinate in the hyperbolic geometry \cite{de2000holographic}, which plays the role of RG time. The canard-like phenomenon corresponds to the instanton-like jump between the opposite points of the limit cycle. The non-perturbative ``RG instanton'' along the imaginary RG time presumably corresponds to the effective tunneling between the IR and UV domains. The low-dimensional QM/AdS$_2$ holography hopefully provides a suitable environment for investigating this phenomenon.

\item The IQHE effect is a familiar example when the plateau regimes correspond to the quantized first Chern class for the Berry curvature in the parameter space. Our study suggests a more general view on this issue. The dynamical system on the torus considered as the particular example of RG flow  enjoys the regions with the quantized Poincar\'{e} rotation number -- these are exactly the plateaus on the Cantor staircase. It would be interesting to elaborate if the quantized Poincar\'{e} rotation number substitutes the quantized first Chern number for the Berry curvature in the description of the plateau in general RG flow. This line of reasoning is supported by our identification of the Poincar\'{e} rotation number with the Hanney angle in the parametric oscillator picture. Recall that the Hanney angle is the classical counterpart of the Berry phase in the QM.

\item Using the mapping of the dynamical system on the torus to the quantum mechanical problem, we have argued that the canard-like phenomena most probably can be interpreted as a complex time instanton upon the complexification of the slow manifold in the slow-fast limit. On the other hand, they correspond to the conventional instantons in QM hence the summation of the instantons and generic resurgence of the series in $\hbar$ , see, for instance, \cite{cherman2015decoding,basar2013resurgence,aniceto2011resurgence,codesido2019non} can be developed in the dynamical system as well.

\item Generally, the phase-locking domains corresponds to the rational rotation numbers. In the RG context, the fractional plateau is well known for the FQHE where the fractional Cantors staircase for $\sigma_{xy}$ emerges from the interacting quasiparticles instead of the non-interacting quasiparticles for IQHE. It would be interesting to recognize the origin of the fractional rotation number in the framework of the Hill operator and coadjoint Virasoro orbits.

\item In IQHE and other examples \cite{flack2023generalized} the variable that exhibits quantized plateau behavior is related to Berry's curvature,
which classical counterpart we have connected to the rotation number. The Berry curvature is known to be the imaginary part of the quantum tensor, while the real part of the quantum tensor defines the quantum metrics on the parameter space \cite{provost1980riemannian}. It would be interesting to recognize the partner of the Poincar\'{e} rotation number that together forms the modular parameter. The related question concerns the role of the Zak phase \cite{zak1989berry} in the analysis considered.

\item As noted in \cite{klimenko2013asymptotic}, the boundaries of phase-locking domains can be approximated by the Bessel functions in the large $A$ limit. This implies that the constrictions correspond to the zeros of Bessel functions. This observation is interesting from the Virasoro group point of view. Extending the discussion of the Virasoro group, it is known that any symplectic manifold involving the Virasoro coadjoint orbit can be quantized via geometric Kirillov-Kostant quantization. The orbits $\mathrm{Diff}(S^1)/S^1$ yield the Verma module of the Virasoro algebra,  $\mathrm{Diff}(S^1)/\mathrm{SL}^n(2,\mathbb{R})$ yields the Verma module with null vector at $n$-the level. Quantization of the $T_{n,\Delta}$, $T_{n,\pm}$ orbits is not yet known. Hence in general setup with the Cantor staircase picture we expect that the candidate theories for description of plateau are the quantization of the special hyperbolic Virasoro orbits $T_{\delta,n}$ for the $n$-th step on the staircase. On the other hand, the candidate theories for transitions between the steps on the staircase, for instance, in the IQHE case, are described by the CFT, which are quantization of $T_{\alpha,0}$ orbits and $T_{\pm,n}$ at the boundaries of the phase-locking domains. 

\item The constriction points are related to the zeros of the Bessel function corresponding to the fixed values of $A$ in the RSJ model. Therefore, the representatives of the Virasoro orbits are also related to the zeros of the Bessel function. We know that upon quantization the representatives yield the conformal weights for the particular operators. On the other hand, in the holographic picture the zeros of the Bessel functions indeed are related to the conformal weights or
equivalently to the eigenvalues of the bulk equations of motion. It would be interesting to explore the holographic viewpoint on the constriction points.

\item Finally, the dynamical systems on the torus can be linked to the different conformal blocks in the Liouville theory via the Heun equation. On the other hand, such conformal blocks are related via AGT correspondence \cite{alday2010liouville} to the Nekrasov partition functions for SUSY YM theories with different number of flavors; see, for instance, \cite{litvinov2014classical,lisovyy2022perturbative,zenkevich2011nekrasov}. The same Heun equations describe also the quasinormal modes around the black holes \cite{aminov2022black,aminov2023black,bonelli2022exact,bonelli2023irregular}. We will discuss the interplay between the dynamical systems on the torus, $\mathcal{N}=2$ SQCD and quasinormal modes around the black holes in a separate publication \cite{ags}.

\end{itemize}

\section{Conclusion}\label{sec:Conclusion}

In this study, we elaborate the phase-locking phenomenon in the nonlinear classical dynamical system on 2D torus by mapping it to the properties of the Hill operator which provides either the equation of motion for the parametric oscillator, or the Schr\"{o}dinger equation. In the parametric oscillator interpretation, irrational  Poincar\'{e} rotation number is directly related to the nonadiabatic Hannay angle. On the other hand, interpretation in terms of the Schr\"{o}dinger equation, allowed us to obtain a framework for the derivation of the positions and the widths of the phase-locking areas. Since the eigenfunctions of the Hill operator govern the classification of the coadjoint orbits, we derive the classification of the phase-locking domains and its boundaries in terms of the Virasoro coadjoint orbits. The so-called special coadjoint orbits parameterized by a pair of parameters are the key place. An integer parameter is identified with the Poincar\'{e} rotation number, while the second one is an analogue of the Hannay angle in the hyperbolic monodromy region.

The mapping between the quantum mechanics and dynamical systems opens the possibility of applying the conventional quantum mechanical tools to quantify the different phenomena in the dynamical systems treated equivalently as examples of the RG flows. We have shown that the quantization of the Poincar\'e rotation number corresponds to the exact WKB quantization in quantum mechanics with the periodic potential. The semiclassical approximation corresponds to the slow-fast dynamical system, and we have argued that the ``RG instantons'' or cannards can be formulated in terms of the complexified slow manifold of a dynamical system. This viewpoint provides the possibility to formulate the universal RG viewpoint for several phenomena in which Heun-type equations emerge naturally, suggesting some additional tool to develop a holographic picture for them.

\section*{Acknowledgements}

The authors thank Victor Buchshtaber and Anton Gerasimov for useful discussions. Alexander Gorsky thanks the Institut Mittag-Leffler for hospitality and support during the program ``Cohomological aspects of the quantum field theory''. Artem Alexandrov was supported by the Foundation for the Advancement of Theoretical Physics and Mathematics ``BASIS'' (grant No. 23-1-5-41-1) and the IDEAS Research Center. Artem Alexandrov is grateful to Yakov Fominov for giving attention to the literature related to Shapiro steps and to Ivan Mamaev for pointing out the appearance of staircase structures in mechanical problems. Alexey Glutsyuk was supported by the Foundation for the Advancement of Theoretical Physics and Mathematics ``BASIS'' (grant No. 24-7-1-15-1) and by the MSHE project No. FSMG-2024-0048.

\section*{Numerical simulation details}

All numerical simulations were performed in the {\julia} programming language. For a numerical solution of differential equations, we use the package \verb|DifferentialEquations|. For Hamiltonian dynamics, the symplectic integrators were used. For the computation of Lyapunov exponents, we use the approach described in~\autoref{app:Layp-Rho} and use the default integrator.

\appendix

\section{Rotation number \& number of zeros}\label{app:Hill-Rho}

Here we provide the proof for the theorem formulated in~\cite{johnson1982rotation} that relates the number of Hill equation solutions zeros and the rotation number of a dynamical system on torus (see~\autoref{subsec:DynTorusHill}). Its proof is similar the one given in~\cite{johnson1982rotation}.

\begin{theorem} \label{th:Hill-zeros}
    Consider a Hill equation of the form: 
    \begin{equation}\label{eq:app-Hill-general}
        \ddot{u}+g_1(\tau)\dot{u}+g_2(\tau)u=0,
    \end{equation}
    where $g_1(\tau)$ and $g_2(\tau)$ are real-valued $2\pi$-periodic functions. Consider the corresponding dynamical system,
    \begin{equation}\label{eq:app-Hill-torus-dynsys}
        \frac{d\xi}{d\tau}=\frac{g_2(\tau)+1}2+\frac{g_2(\tau)-1}2\cos2\xi-\frac{g_1(\tau)}2\sin2\xi.
    \end{equation}
    on the torus $\mathbb T^2$. Then for every solution $u(\tau)$ of the Hill equation, the number $N(T)$ of its zeros on the interval $[0,T)$ is related to the rotation number of the dynamical system on torus~\eqref{eq:app-Hill-torus-dynsys} as
    \begin{equation}
        \rho = \pi\lim\limits_{T\rightarrow\infty}\frac{N(T)}{T}.
    \end{equation}
\end{theorem}
We can represent the eq.~\eqref{eq:app-Hill-general} as linear system,
\begin{equation}\label{eq:app-Hill-linearsys}
\begin{cases}
    \dot{u} = v, \\
    \dot{v} = -g_2(\tau)u-g_1(\tau)v.
\end{cases}
\end{equation}
Let us give an equivalent statement concerning the relation of the rotation number and zeros of solutions to the Hill equations, in terms of projectivization, using the symmetry $(u,v)\mapsto(-u,-v)$, or in other terms, the symmetry $\xi\mapsto\xi+\pi$ of the corresponding dynamical system on $\mathbb T^2$. Passing to projectivization of the system~\eqref{eq:app-Hill-linearsys} via tautological projection $\mathbb{R}^2_{u,v}\to\mathbb{RP}^1_{[u:v]}$ and treating $\mathbb{RP}^1$ as a circle equipped with the coordinate $\zeta=2\xi$ produces another dynamical system on the quotient torus $\mathbb T^2=S^1_\zeta\times S^1_{\tau}$. The rotation number of the latter dynamical system is clearly equal to $\wt\rho=2\rho$. Theorem \ref{th:Hill-zeros} reformulated for the projectivization states that  
\begin{equation}\label{eq:app-Rho-Limit}
    \wt\rho=2\pi\lim_{T\to+\infty}\frac{N(T)}T
\end{equation}

\begin{proof} {\bf of Theorem \ref{th:Hill-zeros}.} The first step of the proof is the following proposition.

\begin{proposition}\label{prop:Hill-Zeros}
Every two neighbor zeros $\tau_1<\tau_2$ of a solution $u(\tau)$ of the Hill equation~\eqref{eq:app-Hill-general} correspond to two neighbor intersection points of the solution curve $(u(\tau),v(\tau))_{\tau>0}$ of the corresponding system~\eqref{eq:app-Hill-linearsys} with the $v$ axis, and their $v$ coordinates are opposite. The increment in the argument $\xi(\tau_2)-\xi(\tau_1)$ between the two consecutive zeros is equal to $\pi$.
\end{proposition}

\begin{proof} At two neighbor zeros of a solution $u(\tau)$ its derivatives obviously have different signs (Rolle-like theorem). Consider a zero where the derivative is positive, i.e., $u=0$ and  $v>0$. Its projection to the $\xi$-circle is its intersection  with the negative ordinate semi-axis, i.e., the point $\xi=-\pi/2$. The projection to the $\xi$-circle of the orbit of the dynamical system~\eqref{eq:app-Hill-torus-dynsys} starting at the corresponding point goes to the right, since $\dot{u}>0$ and $\xi=\arg(u-iv)=\arg(u-i\dot{u})$. Similarly, each zero where the derivative is negative corresponds to intersection with the upper ordinate semi-axis and going to the left. Thus, the increment of $\xi$ between each two neighbor zeros is equal to $\pi$. This implies the first statement of the proposition, which in its turn implies the second.
\end{proof}

Consider the definition of the rotation number as the limit, as $T\to+\infty$, of the mean argument increment
\begin{equation*}
    \rho_T=\frac{\xi(T)-\xi(0)}{T}
\end{equation*}
along a solution $(u(\tau),v(\tau))$ of system~\eqref{eq:app-Hill-linearsys}. Then the number of intersections of the solution curve $(u(\tau),v(\tau))$, $\tau\in[0,T]$, with the $v$-axis (i.e., the number of zeros of $u(\tau)$ in $[0,T]$) is equal to $\left[(\xi(T)-\xi(0))/\pi\right]$ up to adding
a number of moduli not greater than $2$. This follows from Proposition~\ref{prop:Hill-Zeros}. 
Therefore,
\begin{equation*}
    \rho=\lim_{T\rightarrow+\infty}\frac{\xi(T)-\xi(0)}T=\pi\lim_{T\rightarrow+\infty}\frac{N(T)}T.
\end{equation*}
This proves~\eqref{eq:app-Rho-Limit}.
\end{proof}

\section{Phase-locking and Lyapunov exponents}\label{app:Layp-Rho}

Let us remind that we start with the following Hill equation,
\begin{equation}
    \ddot{u}+g_1(t)\dot{u}+g_2(t)u=0,
\end{equation}
where $g_1(t)$ \& $g_2(t)$ are periodic functions with period $2\pi$. Standard radial projection to the unit circle gives the dynamical system~\eqref{eq:app-torus-dyn} on torus with phase variable $\xi(t)=\arg(u(t)+i\dot{u}(t))$: 
\begin{equation}\label{eq:app-torus-dyn}
    \dot{\xi} =  \frac{1-\cos 2\xi}{2} + \frac{g_2}{2}\left(\cos 2\xi+1\right) - \frac{g_1}{2}\sin 2\xi.
\end{equation}   
Then we introduce the variable $\phi=2\xi$ and get 
\begin{equation}
    \dot{\phi} = \left(g_2+1\right) + \left(g_2-1\right)\cos\phi - g_1\sin\phi.
\end{equation}
Now we pass to complex phase  by setting $\Phi=\exp(i\phi)$: 
\begin{equation}\label{eq:app-torus-dyn-Phi}
    \dot{\Phi} = i\left(g_2+1\right)\Phi + \frac{i(g_2-1)}{2}\left(\Phi^2+1\right) - \frac{g_1}{2}\left(\Phi^2-1\right).
\end{equation}
It is easy to see that the dynamics of the complex phase $\Phi=\Phi(t)$ is nothing more than a family of homeomorphisms of the unit circle $S^1\rightarrow S^1$. This means that such dynamics can be represented by the M\"{o}bius transformation that preserves the unit disk with its boundary. This M\"{o}bius transformation can be represented as
\begin{equation}\label{eq:mobius-transform}
    \mathcal{M}(\Phi) = \zeta\frac{\Phi-w}{1-\overline{w}\Phi},\quad |\zeta|=1,\,|w|<1.
\end{equation}
Such M\"{o}bius transformations form 3D Lie group $G\simeq \mathrm{PSL}(2,\mathbb{R})$. There are three canonical one-parametric families of these transformations that can be defined by setting two of the three parameters to zero. Differentiation of these families gives infinitesimal generators of $G$,
\begin{equation}
    v_1 = i\Phi,\quad v_2=\Phi^2-1,\quad v_3=i\Phi^2+i.
\end{equation}
We see that the right-hand side of Eq.~\eqref{eq:app-torus-dyn} is the linear combination of these generators with coefficients that depend only on time. It means that we can write down the dynamics as
\begin{equation}\label{eq:mobius-dynamics}
    \Phi(t) = \zeta(t)\frac{\Phi_0-w(t)}{1-\overline{w}(t)\Phi_0},\quad \Phi_0=\exp(i\phi_0).
\end{equation}
Computing the time derivative and using the expression for inverse M\"{o}bius transformation, we find
\begin{equation}\label{eq:mobius-dyn-der}
    \dot{\Phi}=-\frac{\dot{w}\zeta}{1-|w|^2}+\left(\dot{\zeta}\overline{\zeta}+\frac{\dot{\overline{w}}w-\dot{w}\overline{w}}{1-|w|^2}\right)\Phi+\frac{\dot{\overline{w}}\overline{\zeta}}{1-|w|^2}\Phi^2.
\end{equation}
Indeed, differentiating~\eqref{eq:mobius-dynamics} in time yields
\begin{multline}\label{eq:mob-dyn-smplf}
    \dot{\Phi}=\dot{\zeta}\frac{\Phi_0-w}{1-\overline{w}\Phi_0}+\zeta\left(-\frac{\dot{w}}{1-\overline{w}\Phi_0}\right)-\zeta\left(\frac{\Phi_0-w}{(1-\overline{w}\Phi_0)^2}\right)\left(-\dot{\overline{w}}\Phi_0\right)=\\=\dot{\zeta}\overline{\zeta}\Phi-\left(\frac{\dot{w}}{\Phi_0-w}\right)\Phi+\overline{\zeta}\dot{\overline{w}}\left(1+\frac{w}{\Phi_0-w}\right)\Phi^2.
\end{multline}
Then we invert the relation~\eqref{eq:mobius-transform},
\begin{equation*}
    \Phi_0=\frac{\Phi+\zeta w}{\zeta+\overline{w}\Phi},\quad \frac{1}{\Phi_0-w}=\frac{\zeta+\overline{w}\Phi}{\Phi\left(1-|w|^2\right)}
\end{equation*}
Expressing the terms in~\eqref{eq:mobius-dyn-der} by using the above formulas yields
\begin{equation}\label{eq:mobius-exp-simplified}
\begin{gathered}
    -\left(\frac{\dot{w}}{\Phi_0-w}\right)\Phi=-\dot{w}\Phi\cdot\frac{\zeta+\overline{w}\Phi}{\Phi(1-|w|^2)}=-\frac{\dot{w}\zeta}{1-|w|^2}-\frac{\dot{w}\overline{w}}{(1-|w|^2)},\\
    \frac{\overline{\zeta}\dot{\overline{w}}w}{\Phi_0-w}\Phi^2=\overline{\zeta}\dot{\overline{w}}w\Phi^2\cdot \frac{\zeta+\overline{w}\Phi}{\Phi(1-|w|^2)}=\frac{\dot{\overline{w}}w}{1-|w|^2}\Phi+\frac{\overline{\zeta}|w|^2\dot{\overline{w}}}{1-|w|^2}\Phi^2,\\
    \overline{\zeta}\dot{\overline{w}}\left(1+\frac{w}{\Phi_0-w}\right)\Phi^2=\frac{\dot{\overline{w}}w}{1-|w|^2}\Phi+\frac{\dot{\overline{w}}\overline{\zeta}}{1-|w|^2}\Phi^2.
\end{gathered}
\end{equation}
Substituting the formulas from~\eqref{eq:mobius-exp-simplified} into equation~\eqref{eq:mob-dyn-smplf} yields~\eqref{eq:mobius-dyn-der}. Next, it is convenient to define the functions
\begin{equation}
    g_{+}(t) = \frac{g_1(t)}{2}+\frac{i(g_2(t)-1)}{2},\quad g_{-}(t) = -\frac{g_1(t)}{2}+\frac{i(g_2(t)-1)}{2}.
\end{equation}
Matching the right hand sides of  equations~\eqref{eq:app-torus-dyn-Phi} and~\eqref{eq:mobius-dyn-der} we find
\begin{equation}
    \frac{\dot{w}}{1-|w|^2} = -g_{+}\overline{\zeta},\quad \frac{\dot{\overline{w}}}{1-|w|^2} = g_{-}\zeta,
\end{equation}
\begin{equation}
    \dot{\zeta}=i(g_2+1)\zeta-g_{-}w\zeta^2-g_{+}\overline{w}.
\end{equation}
Therefore, we have system of two equations for parameters $w=w(t)$ and $\zeta=\zeta(t)$,
\begin{equation}
    \dot{w}=-\left(1-|w|^2\right)g_{+}\bar\zeta,\quad \dot{\zeta}=-i(g_2+1)\zeta-g_{-}w\zeta^2-g_{+}\overline{w}.
\end{equation}
The initial conditions for these equations are known. In fact, from Eq.~\eqref{eq:mobius-dynamics} we have $w(0)=0$, $\zeta(0)=1$. With knowledge about the M\"{o}bius dynamics, we can write down the following expression for the Poincar\'{e} map,
\begin{equation}
    \phi(2\pi)=\im\log\mathcal{P}\left(e^{i\phi_0}\right),\quad \mathcal{P}(\bullet)=\zeta(2\pi)\frac{\bullet-w(2\pi)}{1-\overline{w}(2\pi)\bullet}.
\end{equation}
This means that the Poincar\'{e} map is nothing more than the M\"{o}bius transformation with time-dependent coefficients taken at time $t=2\pi$. We conclude that the Poincar\'{e} map defines the discrete dynamical system,
\begin{equation}
    \Phi_{n+1}=\mathcal{P}(\Phi_n).
\end{equation}
By definition the Lyapunov exponent for discrete dynamical system is
\begin{equation}
    \Lambda = \lim\limits_{n\rightarrow\infty}\frac{1}{n}\sum_{i=0}^{n-1}\ln\left|\mathcal{P}'(\Phi_i)\right|=\ln\left|\mathcal{P}'\left(\Phi^*\right)\right|,
\end{equation}
where $\Phi^*$ is the attracting (or neutral in the parabolic case) fixed point of map $\mathcal{P}$. The above limit exists for every initial condition $\Phi_0$. 

In the theory of M\"{o}bius transformations, the derivatives of a M\"{o}bius transformation at fixed points are called its multipliers. The general M\"{o}bius transformation can have two fixed points, so we generally have two multipliers, $\mu_1$ and $\mu_2$. The multipliers obey the relation $\mu_1\mu_2=1$ and can be expressed through eigenvalues $\lambda_1$, $\lambda_2$ of  M\"{o}bius transformation matrix: $\mu_1=\lambda_2/\lambda_1$, $\mu_2=\mu_1^{-1}$.
The M\"{o}bius transformations preserving the unit disk and the unit circle can be classified by the trace of the M\"{o}bius transformation matrix $M$ normalized by scalar factor to have unit determinant, which has the form:
\begin{equation}
    M=\frac{1}{\sqrt{|a|^2-|b|^2}}\begin{pmatrix}a & b \\ \overline{b} & \overline{a} \end{pmatrix},\quad a,\,b\in\mathbb{C}, \ \ |a|^2-|b|^2>0.
\end{equation}
The latter matrix has real trace. If $0<|\tr M|<2$, the transformation is called \emph{elliptic}, its fixed points lie outside the unit circle, the multipliers are conjugated, $\mu_1=\overline{\mu_2}$, and the transformation is conjugated to the rotation of the unit circle. If $|\tr M|>2$, the transformation is called \emph{hyperbolic}, it has two fixed points on the unit circle, the multipliers are real, $\mu_1>1$ and $\mu_2<1$. If $|\tr M|=2$ and the transformation is not identity, it is called \emph{parabolic}, it has one fixed point on the unit circle, and $\mu_1=\mu_2=1$.

We perform the following procedure. First of all, the normalized transformation $\mathcal{P}$ has the following matrix form,
\begin{equation}
    M = \frac{1}{\sqrt{\zeta(1-|w|^2)}}\begin{pmatrix}\zeta & -w\zeta \\ 
    -\overline{w} & 1\end{pmatrix}
\end{equation}
with following eigenvalues,
\begin{equation}
    \lambda_{\pm}=\frac{1+\zeta\pm\sqrt{(1-\zeta)^2+4\zeta|w|^2}}{2\sqrt{\zeta(1-|w|^2)}},
\end{equation}
which gives the explicit expressions for multipliers $\mu_1=\mu$, $\mu_2=1/\mu$
\begin{equation}
    \mu = \frac{1+\zeta+\sqrt{(1-\zeta)^2+4\zeta|w|^2}}{1+\zeta-\sqrt{(1+\zeta)^2+4\zeta|w|^2}}.
\end{equation}
Second, we numerically solve the equations for $\zeta=\zeta(t)$ and $w=w(t)$ and find $\zeta(2\pi)$ and $w(2\pi)$. It allows us compute multipliers $\mu_1$ \& $\mu_2$ as the functions of parameters $A$, $B$ and $\omega$. Finally, we compute the Lyapunov exponent $\Lambda=\ln|\mu|$. This explains why in the case of RSJ model the phase-locking occurs for the hyperbolic element of $\mathrm{SL}(2,\mathbb{R})$: in such a case the non-trivial fixed point on the unit circle $S^1$ exists. So, in the sense of the RSJ model, the hyperbolic monodromy of the Hill equation is stable, which means that the phase-locking occurs.

\section{Invariant combinations of Hill equation solutions}\label{app:Inv-Virasoro}

Consider the following Hill equation,
\begin{equation}\label{eq:Hill-Appendix}
    \ddot{u} + V(t)u = 0,\quad V(t+T)=V(t).
\end{equation}
The Wronskian being time-independent, the monodromy is unimodular, i.e., it has unit determinant. We consider that the monodromy is nonparabolic (i.e. not a Jordan cell). Hence, it is diagonalizable with complex inverse eigenvalues. This implies that it preserves a non-degenerate quadratic form, which is given by an off-diagonal matrix in the (real or complex) basis diagonalizing the monodromy.   This means that for every pair of linearly independent solutions $u_1(t)$ and $u_2(t)$ of the eq.~\eqref{eq:Hill-Appendix} there exists a symmetric non-degenerate matrix  $A\in \mathrm{GL}(2,\mathbb{C})$ (real if $V$ is real-valued,  
complex if  $V$ is complex-valued), set
\begin{equation}
    Q(t)=\left\langle AU(t),U(t)\right\rangle, \ \ U(t)=(u_1(t),u_2(t)),
\end{equation}
such that $\forall t$ $Q(t+T)=Q(t)$. Rescale $u_1$ and $u_2$ by one and the same complex constant factor to obtain basic solutions with unit Wronskian. This fixes the solution basis up to the unimodular linear transformation.

Let us choose a system of linearly independent complex  solutions $u_1(t)$, $u_2(t)$ of~\eqref{eq:Hill-Appendix} diagonalizing the complexified monodromy, normalized to have the unit Wronskian,
\begin{equation}
    u_1(t+T)=e^{\mu}u_1(t),\quad u_2(t+T)=e^{-\mu}u_2(t),\,\,W=u_1\dot{u}_2-u_2\dot{u}_1=1,
\end{equation}
where $\mu\in\mathbb{C}$. Then the quadratic form is given by $u_1u_2$. The following theorem takes place.
\begin{theorem}\label{th:Hill-invariants}
    In the above condition, the exponent $\mu$ can be found by formula \begin{equation}\label{eq:Lambda-Exp-appendix}
        \mu = -\frac{1}{2}\left(\mathrm{p.v.}\,\,\int_{t_0}^{t_0+T}\frac{ds}{u_1(s)u_2(s)}\right).
    \end{equation}
\end{theorem}

\begin{proof}
    First of all, let us make the following remark. The solutions $u_1$ and $u_2$ may have zeros in the integration interval, it explains the necessity of taking the principal value of the integral. The above integral is an integral of the $T$-periodic function over a period. It is independent on $t_0$ and is treated as a principal value integral over the circle $\mathbb{R}/(T\mathbb {Z})$.

    Let us write $u_2(t)=e^{\phi(t)}u_1(t)$, then the Wronskian is given by
    \begin{equation*}
        W = u_1\dot{u}_2-u_2\dot{u}_1 = e^{\phi(t)}\dot{\phi}(t)u_1^2 = \dot{\phi}(t)u_1(t)u_2(t)=1.
    \end{equation*}
    Performing the integration, we find
    \begin{equation}
        \phi(t) = \phi(t_0) + \int_{t_0}^{t}\frac{ds}{u_1(s)u_2(s)}.
    \end{equation}
    One has
    \begin{equation*}
        u_2(t+T) = e^{\phi(t+T)}u_1(t+T)=e^{\mu+\phi(t+T)}u_1(t) = e^{-\mu}u_2(t) = e^{-\mu+\phi(t)}u_1(t),
    \end{equation*}
    therefore
    \begin{equation*}
        \phi(t+T)-\phi(t)=-2\mu.
    \end{equation*}
    This implies the formula~\eqref{eq:Lambda-Exp-appendix} when $u_1$ and $u_2$ have no zeros.

     Let us now consider the case where one of the solutions, say $u_1$, has zeros. Then its zeros are simple and form a $T$-periodic lattice, since $u_1$ is a monodromy eigenfunction. Therefore, $u_2$ also has zeros intermittent with those of $u_1$ (Sturm Theorem), simple and similarly, they form a $T$-periodic lattice. Denote the zeros of $u_1$ by $a_1,\dots,a_k$, and those of $u_2$ by $b_1,\dots,b_k$, and one has $a_1<b_1<a_2<b_2<\dots$ (up to the exchange of $b_j$ and $a_j$). In this case, the above argument applies along a path $\Gamma$ in $\mathbb{C}_t$ from $t_0$ to $t_0+T$ that avoids zeros of $u_1$, $u_2$. It goes along the real line to the right until it approaches a zero $c$ of some $u_l$, then goes along the semicircle of small radius centered at $c$ in the upper half-plane, then goes along the real line to the right, etc.

    The above argument shows that $\mu$ is equal to the same integral as in~\eqref{eq:Lambda-Exp-appendix}, but taken along the path $\Gamma$ and without a p.v. sign. As the radii of the above semicircles tend to zero, the integral does not change and, on the other hand, tends to the principal value integral in~\eqref{eq:Lambda-Exp-appendix} plus $\pi i$ times the sum of the residues of the form $(u_1(t)u_2(t))^{-1}dt$ at the zeros; the form being taken in a small tubular neighborhood of the circle $S^1=\mathbb{R}/(T\mathbb{Z})$ in the complex domain. At each zero $a_j$ the residue is equal to $-1$, since $\dot{u}_1(a_j)u_2(a_j)=-W=-1$. Similarly, at each zero $b_j$ the residue is equal to 1, since $u_1(b_j)\dot{u}_2(b_j)=W=1$. Finally, the residues cancel out and the integral along the path $\alpha$ is equal to the principle value integral~\eqref{eq:Lambda-Exp-appendix}.
\end{proof}
Consider now the elliptic case, where the monodromy is conjugated to a rotation. We consider that the monodromy is a rotation in some basis of solutions $u_1(t)$ and $u_2(t)$ with unit Wronskian. Thus, the quadratic form $u_1^2(t)+u_2^2(t)$ is $T$-periodic: monodromy invariant. In this case, the complexified monodromy operator diagonalizes in the basis $v_1=(iu_1-u_2)/\sqrt{2}$, $v_2=(u_1-iu_2)/\sqrt{2}$
and has eigenvalues $e^{i\alpha}$ and $e^{-i\alpha}$ respectively, where $\alpha$ is the rotation angle. The basis of solutions $v_1$, $v_2$ has a unit Wronskian, and the function $v_1v_2$ has no zeros, since linearly independent solutions $u_1$ and $u_2$ have no common zeros. The theorem gives us the expression for $\alpha$,
\begin{equation}
    \alpha = \int_{0}^{T}\frac{ds}{u_1(s)^2+u_2(s)^2},
\end{equation}
so the eigenvalues $\lambda_{\pm}$ of the monodromy matrix are given by
\begin{equation}
    \lambda_{\pm} = e^{\pm i\alpha} = \exp\left\{\pm i\int_{0}^{T}\frac{ds}{u_1(s)^2+u_2(s)^2}\right\}.
\end{equation}
In the hyperbolic case, we have
\begin{equation}
    \lambda_{\pm} = e^{\pm \lambda} = \exp\left\{\mp\frac{1}{2}\int_{0}^{T}\frac{ds}{u_1(s)u_2(s)}\right\}
\end{equation}

\section{Hill equation for RSJ model}\label{app:Hill-RSJ}

Here we explicitly show how to obtain the Hill equation from the RSJ model, which was done in \cite{renne1974some}. The phase dynamics in the RSJ model is given by
\begin{equation}
    \dot{\phi} = -\sin\phi + B + A\cos\omega t.
\end{equation}
To find the corresponding Hill equation, we make the series of the following variable changes. First, change the definition of phase by setting $\theta=\phi-\pi/2$,
\begin{equation}
    \frac{d\theta}{dt}+\cos\theta=B+A\cos\omega t.
\end{equation}
Then, introduce the variable $x=\tan(\theta/2)$, which gives the following equation,
\begin{equation}\label{eq:RSJ-Riccati-App}
    \frac{dx}{dt} = f_{+}(t)x^2+f_{-}(t),\quad f_{\pm}=\frac{1}{2}\left(\pm 1+B+A\cos\omega t\right).
\end{equation}
All the solutions of the Riccati equation~\eqref{eq:RSJ-Riccati-App} have the form 
$x=Y_2/Y_1$, where $Y=(Y_1(t),\,Y_2(t))^T$ is a vector solution of linear system
\begin{equation*}
    \frac{d}{dt}\begin{pmatrix}
        Y_1 \\ Y_2
    \end{pmatrix} = 
    \begin{pmatrix} 0 & -f_+(t)\\ f_-(t) & 0\end{pmatrix}
    \begin{pmatrix}
        Y_1 \\ Y_2
    \end{pmatrix}.
\end{equation*}
Linear change $Y_2\mapsto\wt Y_2=-f_+(t)Y_2$ transforms the above system to the system
\begin{equation*}
    \frac{d}{dt}\left(\begin{matrix} Y_1 \\ \wt Y_2\end{matrix}\right)
    =\left(\begin{matrix} 0 & 1 \\
    -f_+(t)f_-(t) & d(\ln f_+(t))/dt \end{matrix}\right)\left(\begin{matrix} Y_1\\ \wt Y_2\end{matrix}\right),
\end{equation*}
which is equivalent to the second order differential equation
\begin{equation*}
    \frac{d^2Y_1}{dt^2}-\frac{d\ln f_+(t)}{dt}\frac{dY_1}{dt}+f_+(t)f_-(t)Y_1=0.
\end{equation*}
The gauge transformation $Y_1(t)\mapsto u(t)=Y_1(t)/\sqrt{f_+(t)}$ sends it to the Hill equation
\begin{equation}
    \frac{d^2u}{dt}+V(t)u(t)=0,
\end{equation}
where the potential $V=V(t)$ is given by
\begin{equation}
   V(t)=f_{+}(t)f_{-}(t)-\frac{3}{4}\left(\frac{1}{f_{+}}\frac{df_{+}}{dt}\right)^2+\frac{1}{2f_{+}}\frac{d^2f_{+}}{dt^2}.
\end{equation}
To find explicit form of the potential, let us calculate its second and third terms denoting
\begin{equation*}
    \frac{3}{4}\left(\frac{\dot{f}_{+}(t)}{f_+(t)}\right)^2=
    \frac34\frac{\omega^2a^2\sin^2(\omega t)}{(1+a\cos(\omega t))^2},\ \ 
    \frac{\ddot{f}_{+}(t)}{2f_+(t)}=-\frac{\omega^2a\cos(\omega t)}{2(1+a\cos(\omega t))},\quad a=\frac{A}{B+1}.
\end{equation*}
This yields 
\begin{equation}
  V(t) = \frac{1}{4}\left\{(B+A\cos\omega t)^2-1\right\}-\frac{\omega^2}{4}\frac{2a^2+2a\cos\omega t + a^2\sin^2\omega t}{(1+a\cos\omega t)^2},\quad a=\frac{A}{B+1}.
\end{equation}
This form is convenient to see the limiting cases,
\begin{equation}
   \lim\limits_{A\rightarrow 0}V(t) = \frac{B^2-1}{4},\quad \lim\limits_{\omega\rightarrow 0} V(t) = \frac{(B+A)^2-1}{4}.
\end{equation}
Note that in the limit $A\ll 1$ with $\omega\rightarrow 0$, the potential $V(t)$ becomes
\begin{equation}
   V(t) = \frac{B^2-1}{4}+\frac{AB}{2}\cos\omega t-\frac{A\omega^2}{2(B+1)}+o(A\omega^2).
\end{equation}

{
\footnotesize
\bibliography{ref.bib}

\begin{thebibliography}{100}

\bibitem{pikovsky2001synchronization}
A.~Pikovsky, M.~Rosenblum, and J.~Kurths.
\newblock {\em Synchronization: {A} {U}niversal {C}oncept in {N}onlinear {S}ciences}.
\newblock Cambridge University Press, 2001.
\newblock \href {https://doi.org/10.1017/CBO9780511755743} {\path{doi:10.1017/CBO9780511755743}}.

\bibitem{mccumber1968effect}
D.~E. McCumber.
\newblock Effect of {AC} impedance on {DC} voltage-current characteristics of superconductor weak-link junctions.
\newblock {\em Journal of Applied Physics}, 39(7):3113--3118, 1968.
\newblock \href {https://doi.org/10.1063/1.1656743} {\path{doi:10.1063/1.1656743}}.

\bibitem{mishra2025phase}
S.~Mishra, A.~Ryabov, and P.~Maass.
\newblock Phase locking and fractional {S}hapiro steps in collective dynamics of microparticles.
\newblock {\em Physical Review Letters}, 134(10):107102, 2025.
\newblock \href {https://doi.org/10.1103/PhysRevLett.134.107102} {\path{doi:10.1103/PhysRevLett.134.107102}}.

\bibitem{dinsmore2008}
R.~C. Dinsmore, M.~Bae, and A.~Bezryadin.
\newblock Fractional order {S}hapiro steps in superconducting nanowires.
\newblock {\em Applied physics letters}, 93(19), 2008.
\newblock \href {https://doi.org/10.1063/1.3012360} {\path{doi:10.1063/1.3012360}}.

\bibitem{gruner1985charge}
G.~Gr\"{u}ner and A.~Zettl.
\newblock Charge density wave conduction: A novel collective transport phenomenon in solids.
\newblock {\em Physics Reports}, 119(3):117--232, 1985.
\newblock \href {https://doi.org/10.1016/0370-1573(85)90073-0} {\path{doi:10.1016/0370-1573(85)90073-0}}.

\bibitem{reichhardt2015shapiro}
С. Reichhardt and J.~Olson.
\newblock Shapiro steps for skyrmion motion on a washboard potential with longitudinal and transverse ac drives.
\newblock {\em Physical Review B}, 92(22):224432, 2015.
\newblock \href {https://doi.org/10.1103/PhysRevB.92.22443} {\path{doi:10.1103/PhysRevB.92.22443}}.

\bibitem{thouless1982quantized}
D.~J. Thouless, M.~Kohmoto, M.~P. Nightingale, and M.~den Nijs.
\newblock Quantized {H}all conductance in a two-dimensional periodic potential.
\newblock {\em Physical review letters}, 49(6):405, 1982.
\newblock \href {https://doi.org/10.1103/PhysRevLett.49.405} {\path{doi:10.1103/PhysRevLett.49.405}}.

\bibitem{avron1985quantization}
J.~E. Avron and R.~Seiler.
\newblock Quantization of the {H}all conductance for general, multiparticle {S}chr{\"o}dinger hamiltonians.
\newblock {\em Physical review letters}, 54(4):259, 1985.
\newblock URL: \url{10.1103/PhysRevLett.54.259}, \href {https://doi.org/10.1103/PhysRevLett.54.259} {\path{doi:10.1103/PhysRevLett.54.259}}.

\bibitem{avron1995viscosity}
J.~E. Avron, R.~Seiler, and P.~G. Zograf.
\newblock Viscosity of quantum {H}all fluids.
\newblock {\em Physical review letters}, 75(4):697, 1995.
\newblock \href {https://doi.org/10.1103/PhysRevLett.75.697} {\path{doi:10.1103/PhysRevLett.75.697}}.

\bibitem{flack2023generalized}
A.~Flack, A.~Gorsky, and S.~Nechaev.
\newblock Generalized {D}evil's staircase and {RG} flows.
\newblock {\em Nuclear Physics B}, 996:116376, 2023.
\newblock \href {https://doi.org/10.1016/j.nuclphysb.2023.116376} {\path{doi:10.1016/j.nuclphysb.2023.116376}}.

\bibitem{arn}
V.~Arnold.
\newblock {\em Geometrical methods in the theory of ordinary differential equations}, volume 250.
\newblock Springer Science \& Business Media, 2012.

\bibitem{buchstaber2010rotation}
V.~Buchstaber, O.~Karpov, and S.~Tertychniy.
\newblock Rotation number quantization effect.
\newblock {\em Theoretical and Mathematical Physics}, 162:211--221, 2010.
\newblock \href {https://doi.org/10.1007/s11232-010-0016-4} {\path{doi:10.1007/s11232-010-0016-4}}.

\bibitem{shapiro1963josephson}
S.~Shapiro.
\newblock {J}osephson currents in superconducting tunneling: {T}he effect of microwaves and other observations.
\newblock {\em Physical Review Letters}, 11(2):80, 1963.
\newblock \href {https://doi.org/10.1103/PhysRevLett.11.80} {\path{doi:10.1103/PhysRevLett.11.80}}.

\bibitem{renne1974some}
M.~J. Renne and D.~Polder.
\newblock Some analytical results for the resistively shunted {J}osephson junction.
\newblock {\em Revue de physique appliqu{\'e}e}, 9(1):25--28, 1974.
\newblock \href {https://doi.org/10.1051/rphysap:019740090102500} {\path{doi:10.1051/rphysap:019740090102500}}.

\bibitem{waldram1982alternative}
J.~R. Waldram and P.~H. Wu.
\newblock An alternative analysis of the nonlinear equations of the current-driven {J}osephson junction.
\newblock {\em Journal of Low Temperature Physics}, 47:363--374, 1982.
\newblock \href {https://doi.org/10.1007/BF00683738} {\path{doi:10.1007/BF00683738}}.

\bibitem{buchstaber2017monodromy}
V.~Buchstaber and A.~Glutsyuk.
\newblock On monodromy eigenfunctions of {H}eun equations and boundaries of phase-lock areas in a model of overdamped {J}osephson effect.
\newblock {\em Proceedings of the Steklov Institute of Mathematics}, 297:50--89, 2017.
\newblock \href {https://doi.org/10.1134/S0081543817040046} {\path{doi:10.1134/S0081543817040046}}.

\bibitem{glutsyuk2014adjacency}
A.~Glutsyuk, V.~Kleptsyn, D.~Filimonov, and I.~Schurov.
\newblock On the adjacency quantization in an equation modeling the {J}osephson effect.
\newblock {\em Functional Analysis and Its Applications}, 48(4):272--285, 2014.
\newblock \href {https://doi.org/10.1007/s10688-014-0070-z} {\path{doi:10.1007/s10688-014-0070-z}}.

\bibitem{glutsyuk2019constrictions}
A.~Glutsyuk.
\newblock On constrictions of phase-lock areas in model of overdamped {J}osephson effect and transition matrix of the double-confluent {H}eun equation.
\newblock {\em Journal of Dynamical and Control Systems}, 25(3):323--349, 2019.
\newblock \href {https://doi.org/10.1007/s10883-018-9411-1} {\path{doi:10.1007/s10883-018-9411-1}}.

\bibitem{tertychniy2006}
S.~I. Tertychniy.
\newblock The modelling of a {J}osephson junction and {H}eun polynomials, 2006.
\newblock \href {https://arxiv.org/abs/math-ph/0601064} {\path{arXiv:math-ph/0601064}}.

\bibitem{buchstaber2013explicit}
V.~Buchstaber and S.~Tertychniy.
\newblock Explicit solution family for the equation of the resistively shunted {J}osephson junction model.
\newblock {\em Theoretical and Mathematical Physics}, 176(2):965--986, 2013.
\newblock \href {https://doi.org/10.1007/s11232-013-0085-2} {\path{doi:10.1007/s11232-013-0085-2}}.

\bibitem{buchstaber2015holomorphic}
V.~Buchstaber and S.~Tertychnyi.
\newblock Holomorphic solutions of the double confluent {H}eun equation associated with the {RJS} model of the {J}osephson junction.
\newblock {\em Theoretical and Mathematical Physics}, 182:329--355, 2015.
\newblock \href {https://doi.org/10.1007/s11232-015-0267-1} {\path{doi:10.1007/s11232-015-0267-1}}.

\bibitem{bondeson1985}
A.~Bondeson, E.~Ott, and Thomas~M Antonsen~Jr.
\newblock Quasiperiodically forced damped pendula and {S}chr{\"o}dinger equations with quasiperiodic potentials: implications of their equivalence.
\newblock {\em Physical review letters}, 55(20):2103, 1985.
\newblock \href {https://doi.org/10.1103/PhysRevLett.55.2103} {\path{doi:10.1103/PhysRevLett.55.2103}}.

\bibitem{bizyaev2017hess}
I.~A. Bizyaev, A.~V. Borisov, and I.~S. Mamaev.
\newblock The {H}ess-{A}ppelrot case and quantization of the rotation number.
\newblock {\em Regular and Chaotic Dynamics}, 22:180--196, 2017.
\newblock \href {https://doi.org/10.1134/S156035471702006X} {\path{doi:10.1134/S156035471702006X}}.

\bibitem{johnson1982rotation}
R.~Johnson and J.~Moser.
\newblock The rotation number for almost periodic potentials.
\newblock {\em Communications in Mathematical Physics}, 84(3):403--438, 1982.
\newblock URL: \url{https://doi.org/10.1007/BF01208484}, \href {https://doi.org/doi.org/10.1007/BF01208484} {\path{doi:doi.org/10.1007/BF01208484}}.

\bibitem{brezin1970pair}
E.~Br{\'e}zin and C.~Itzykson.
\newblock Pair production in vacuum by an alternating field.
\newblock {\em Physical Review D}, 2(7):1191, 1970.
\newblock \href {https://doi.org/10.1103/PhysRevD.2.1191} {\path{doi:10.1103/PhysRevD.2.1191}}.

\bibitem{popov1972pair}
V.~S. Popov.
\newblock Pair production in a variable external field (quasiclassical approximation).
\newblock {\em Soviet Journal of Experimental and Theoretical Physics}, 34:709, 1972.

\bibitem{diener1984canard}
M.~Diener.
\newblock The canard unchained or how fast/slow dynamical systems bifurcate.
\newblock {\em The Mathematical Intelligencer}, 6(3):38--49, 1984.
\newblock URL: \url{doi.org/10.1007/BF03024127}, \href {https://doi.org/10.1007/BF03024127} {\path{doi:10.1007/BF03024127}}.

\bibitem{desroches2011canards}
M.~Desroches and M.~R. Jeffrey.
\newblock Canards and curvature: the ``smallness of {$\epsilon$}`` in slow-fast dynamics.
\newblock {\em Proceedings of the Royal Society A: Mathematical, Physical and Engineering Sciences}, 467(2132):2404--2421, 2011.
\newblock \href {https://doi.org/10.1098/rspa.2011.0053} {\path{doi:10.1098/rspa.2011.0053}}.

\bibitem{guckenheimer2001duck}
J.~Guckenheimer and Y.~Ilyashenko.
\newblock The duck and the devil: canards on the staircase.
\newblock {\em Moscow Mathematical Journal}, 1(1):27--47, 2001.
\newblock \href {https://doi.org/10.17323/1609-4514-2001-1-1-27-47} {\path{doi:10.17323/1609-4514-2001-1-1-27-47}}.

\bibitem{shchurov2010canard}
I.~Shchurov.
\newblock Canard cycles in generic fast-slow systems on the torus.
\newblock {\em Transactions of the Moscow Mathematical Society}, 71:175--207, 2010.
\newblock \href {https://doi.org/10.1090/S0077-1554-2010-00184-7} {\path{doi:10.1090/S0077-1554-2010-00184-7}}.

\bibitem{schurov2017duck}
I.~Schurov and N.~Solodovnikov.
\newblock Duck factory on the two-torus: multiple canard cycles without geometric constraints.
\newblock {\em Journal of Dynamical and Control Systems}, 23:481--498, 2017.
\newblock \href {https://doi.org/10.1007/s10883-016-9335-6} {\path{doi:10.1007/s10883-016-9335-6}}.

\bibitem{kleptsyn2013josephson}
V.~Kleptsyn, O.~Romaskevich, and I.~Schurov.
\newblock Josephson effect and slow-fast systems, 2013.
\newblock \href {https://arxiv.org/abs/1305.6755} {\path{arXiv:1305.6755}}.

\bibitem{prufer1926neue}
H.~Pr{\"u}fer.
\newblock Neue herleitung der {S}turm-{L}iouvilleschen {R}eihenentwicklung stetiger funktionen.
\newblock {\em Mathematische Annalen}, 95(1):499--518, 1926.
\newblock \href {https://doi.org/10.1007/BF01206624} {\path{doi:10.1007/BF01206624}}.

\bibitem{broer1995geometrical}
H.~Broer and M.~Levi.
\newblock Geometrical aspects of stability theory for {H}ill's equations.
\newblock {\em Archive for rational mechanics and analysis}, 131:225--240, 1995.
\newblock \href {https://doi.org/10.1007/BF00382887} {\path{doi:10.1007/BF00382887}}.

\bibitem{broer2000resonance}
H.~Broer and C.~Sim{\'o}.
\newblock Resonance tongues in {H}ill's equations: a geometric approach.
\newblock {\em Journal of Differential Equations}, 166(2):290--327, 2000.
\newblock \href {https://doi.org/10.1006/jdeq.2000.3804} {\path{doi:10.1006/jdeq.2000.3804}}.

\bibitem{schon1990quantum}
G.~Sch{\"o}n and A.~Zaikin.
\newblock Quantum coherent effects, phase transitions, and the dissipative dynamics of ultra small tunnel junctions.
\newblock {\em Physics Reports}, 198(5-6):237--412, 1990.
\newblock \href {https://doi.org/10.1016/0370-1573(90)90156-V} {\path{doi:10.1016/0370-1573(90)90156-V}}.

\bibitem{zaikin2019dissipative}
A.~Zaikin and D.~Golubev.
\newblock {\em Dissipative quantum Mechanics of Nanostructures: Electron Transport, Fluctuations, and Interactions}.
\newblock Jenny Stanford Publishing, 2019.
\newblock URL: \url{10.1201/9780429298233}, \href {https://doi.org/10.1201/9780429298233} {\path{doi:10.1201/9780429298233}}.

\bibitem{chikmagalur2024implicit}
K.~Chikmagalur and B.~Bassam.
\newblock An implicit function method for computing the stability boundaries of {H}ill's equation, 2024.
\newblock \href {https://arxiv.org/abs/2408.08390} {\path{arXiv:2408.08390}}.

\bibitem{panghotra2020giant}
R.~Panghotra, B.~Raes, C.~C. de~Souza~Silva, I.~Cools, W.~Keijers, J.~E. Scheerder, V.~V. Moshchalkov, and J.~Van~de Vondel.
\newblock Giant fractional {S}hapiro steps in anisotropic {J}osephson junction arrays.
\newblock {\em Communications Physics}, 3(1):53, 2020.
\newblock \href {https://doi.org/10.1038/s42005-020-0315-5} {\path{doi:10.1038/s42005-020-0315-5}}.

\bibitem{karpov2008modeling}
O.~Karpov, V.~Buchstaber, S.~Tertychniy, J.~Niemeyer, and O.~Kieler.
\newblock Modeling of rf-biased overdamped {J}osephson junctions.
\newblock {\em Journal of Applied Physics}, 104(9), 2008.
\newblock \href {https://doi.org/10.1063/1.3008011} {\path{doi:10.1063/1.3008011}}.

\bibitem{berry1988classical}
M.~V. Berry and J.~H. Hannay.
\newblock Classical non-adiabatic angles.
\newblock {\em Journal of Physics A: Mathematical and General}, 21(6):L325, 1988.
\newblock \href {https://doi.org/10.1088/0305-4470/21/6/002} {\path{doi:10.1088/0305-4470/21/6/002}}.

\bibitem{aharonov1987phase}
Y.~Aharonov and J.~Anandan.
\newblock Phase change during a cyclic quantum evolution.
\newblock {\em Physical Review Letters}, 58(16):1593, 1987.
\newblock \href {https://doi.org/10.1103/PhysRevLett.58.1593} {\path{doi:10.1103/PhysRevLett.58.1593}}.

\bibitem{lewis1968class}
H.~R. Lewis~Jr.
\newblock Class of exact invariants for classical and quantum time-dependent harmonic oscillators.
\newblock {\em Journal of Mathematical Physics}, 9(11):1976--1986, 1968.
\newblock \href {https://doi.org/10.1063/1.1664532} {\path{doi:10.1063/1.1664532}}.

\bibitem{eliezer1976note}
C.~J. Eliezer and A.~Gray.
\newblock A note on the time-dependent harmonic oscillator.
\newblock {\em SIAM Journal on Applied Mathematics}, 30(3):463--468, 1976.
\newblock \href {https://doi.org/10.1137/0130043} {\path{doi:10.1137/0130043}}.

\bibitem{pinney1950nonlinear}
Edmund Pinney.
\newblock The nonlinear differential equation {$y+p(x)y+cy^{-3}=0$}.
\newblock {\em Proc. Amer. Math. Soc}, 1(681):1, 1950.
\newblock \href {https://doi.org/10.1090/S0002-9939-1950-0037979-4} {\path{doi:10.1090/S0002-9939-1950-0037979-4}}.

\bibitem{lazutkin1975normal}
V.~Lazutkin and T.~Pankratova.
\newblock Normal forms and versal deformations for {H}ill's equation.
\newblock {\em Functional Analysis and its applications}, 9(4):306--311, 1975.
\newblock \href {https://doi.org/10.1007/BF01075876} {\path{doi:10.1007/BF01075876}}.

\bibitem{kirillov1981orbits}
A.~Kirillov.
\newblock Orbits of the group of diffeomorphisms of a circle and local {L}ie superalgebras.
\newblock {\em Functional Analysis and Its Applications}, 15(2):135--137, 1981.
\newblock \href {https://doi.org/10.1007/BF01082289} {\path{doi:10.1007/BF01082289}}.

\bibitem{witten1988coadjoint}
E.~Witten.
\newblock Coadjoint orbits of the {V}irasoro group.
\newblock {\em Communications in Mathematical Physics}, 114(1):1--53, 1988.
\newblock \href {https://doi.org/10.1007/BF01218287} {\path{doi:10.1007/BF01218287}}.

\bibitem{balog1998coadjoint}
J.~Balog, L.~Feh\'{e}r, and L.~Palla.
\newblock Coadjoint orbits of the {V}irasoro algebra and the global {L}iouville equation.
\newblock {\em International Journal of Modern Physics A}, 13(02):315--362, 1998.
\newblock \href {https://doi.org/10.1142/S0217751X98000147} {\path{doi:10.1142/S0217751X98000147}}.

\bibitem{blau2024}
M.~Blau and D.~R. Youmans.
\newblock {$\mathrm{SL}(2,\mathbb{R})$} gauge theory, hyperbolic geometry and {V}irasoro coadjoint orbits, 2024.
\newblock \href {https://arxiv.org/abs/2410.01302} {\path{arXiv:2410.01302}}.

\bibitem{klimenko2013asymptotic}
A.~Klimenko and O.~Romaskevich.
\newblock Asymptotic properties of {A}rnold tongues and {J}osephson effect, 2013.
\newblock \href {https://arxiv.org/abs/1305.6746} {\path{arXiv:1305.6746}}.

\bibitem{unterberger2010}
J.~Unterberger.
\newblock A classification of periodic time-dependent generalized harmonic oscillators using a {H}amiltonian action of the {S}chr{\"o}dinger--{V}irasoro group.
\newblock {\em Confluentes Mathematici}, 2(02):217--263, 2010.
\newblock \href {https://doi.org/10.1142/S1793744210000168} {\path{doi:10.1142/S1793744210000168}}.

\bibitem{alekseev1989path}
A.~Alekseev and S.~Shatashvili.
\newblock Path integral quantization of the coadjoint orbits of the {V}irasoro group and 2{D}-gravity.
\newblock {\em Nuclear Physics B}, 323(3):719--733, 1989.
\newblock \href {https://doi.org/10.1016/0550-3213(89)90130-2} {\path{doi:10.1016/0550-3213(89)90130-2}}.

\bibitem{wiegmann1989multivalued}
P.~B. Wiegmann.
\newblock Multivalued functionals and geometrical approach for quantization of relativistic particles and strings.
\newblock {\em Nuclear Physics B}, 323(2):311--329, 1989.
\newblock \href {https://doi.org/10.1016/0550-3213(89)90144-2} {\path{doi:10.1016/0550-3213(89)90144-2}}.

\bibitem{gorsky1991large}
A.~Gorsky, B.~Roy, and K.~Selivanov.
\newblock Large gauge transformations and special orbits of the {V}irasoro group.
\newblock {\em Soviet Journal of Experimental and Theoretical Physics Letters}, 53:64, 1991.

\bibitem{gorsky1995liouville}
A.~Gorsky and A.~Johansen.
\newblock Liouville theory and special coadjoint {V}irasoro orbits.
\newblock {\em International Journal of Modern Physics A}, 10(06):785--799, 1995.
\newblock \href {https://doi.org/10.1142/S0217751X95000371} {\path{doi:10.1142/S0217751X95000371}}.

\bibitem{bazhanov1996integrable}
V.~V. Bazhanov, S.~L. Lukyanov, and A.~B. Zamolodchikov.
\newblock Integrable structure of conformal field theory, quantum {K}d{V} theory and thermodynamic {B}ethe ansatz.
\newblock {\em Communications in Mathematical Physics}, 177(2):381--398, 1996.
\newblock \href {https://doi.org/10.1007/BF02101898} {\path{doi:10.1007/BF02101898}}.

\bibitem{novikov1974periodic}
S.~Novikov.
\newblock The periodic problem for the {K}orteweg-de {V}ries equation.
\newblock {\em Funktsional'nyi Analiz i ego Prilozheniya}, 8(3):54--66, 1974.
\newblock \href {https://doi.org/10.1007/BF01075697} {\path{doi:10.1007/BF01075697}}.

\bibitem{bibilo2022families}
Y.~Bibilo and A.~Glutsyuk.
\newblock On families of constrictions in model of overdamped {J}osephson junction and {P}ainlev{\'e} 3 equation.
\newblock {\em Nonlinearity}, 35(10):5427, 2022.
\newblock \href {https://doi.org/10.1088/1361-6544/ac8aee} {\path{doi:10.1088/1361-6544/ac8aee}}.

\bibitem{voros1983return}
A.~Voros.
\newblock The return of the quartic oscillator. the complex {WKB} method.
\newblock {\em Annales de l'IHP Physique th{\'e}orique}, 39(3):211--338, 1983.
\newblock URL: \url{http://dml.mathdoc.fr/item/AIHPA_1983__39_3_211_0}.

\bibitem{jentschura2004instantons}
U.~D. Jentschura and J.~Zinn-Justin.
\newblock Instantons in quantum mechanics and resurgent expansions.
\newblock {\em Physics Letters B}, 596(1-2):138--144, 2004.
\newblock \href {https://doi.org/10.1016/j.physletb.2004.06.077} {\path{doi:10.1016/j.physletb.2004.06.077}}.

\bibitem{zinn2004multi}
J.~Zinn-Justin and U.~D. Jentschura.
\newblock Multi-instantons and exact results {I}: {C}onjectures, {WKB} expansions, and instanton interactions.
\newblock {\em Annals of Physics}, 313(1):197--267, 2004.
\newblock \href {https://doi.org/10.1016/j.aop.2004.04.004} {\path{doi:10.1016/j.aop.2004.04.004}}.

\bibitem{langer1934asymptotic}
R.~E. Langer.
\newblock The asymptotic solutions of certain linear ordinary differential equations of the second order.
\newblock {\em Transactions of the American Mathematical Society}, 36(1):90--106, 1934.
\newblock \href {https://doi.org/10.2307/1989709} {\path{doi:10.2307/1989709}}.

\bibitem{alvarez2004langer}
G.~Alvarez.
\newblock Langer-{C}herry derivation of the multi-instanton expansion for the symmetric double well.
\newblock {\em Journal of mathematical physics}, 45(8):3095--3108, 2004.
\newblock \href {https://doi.org/10.1063/1.1767988} {\path{doi:10.1063/1.1767988}}.

\bibitem{dunne2014uniform}
G.~V. Dunne and M.~{\"U}nsal.
\newblock Uniform {WKB}, multi-instantons, and resurgent trans-series.
\newblock {\em Physical Review D}, 89(10):105009, 2014.
\newblock \href {https://doi.org/10.1103/PhysRevD.89.105009} {\path{doi:10.1103/PhysRevD.89.105009}}.

\bibitem{milne1930numerical}
W.~E. Milne.
\newblock The numerical determination of characteristic numbers.
\newblock {\em Physical Review}, 35(7):863, 1930.
\newblock \href {https://doi.org/10.1103/PhysRev.35.863} {\path{doi:10.1103/PhysRev.35.863}}.

\bibitem{korsch1985milne}
H.~J. Korsch.
\newblock On {M}ilne's quantum number function.
\newblock {\em Physics Letters A}, 109(7):313--316, 1985.
\newblock \href {https://doi.org/10.1016/0375-9601(85)90181-1} {\path{doi:10.1016/0375-9601(85)90181-1}}.

\bibitem{grassi2020non}
A.~Grassi, J.~Gu, and M.~Mari{\~n}o.
\newblock Non-perturbative approaches to the quantum {S}eiberg-{W}itten curve.
\newblock {\em Journal of High Energy Physics}, 2020(7):1--51, 2020.
\newblock \href {https://doi.org/10.1007/JHEP07(2020)106} {\path{doi:10.1007/JHEP07(2020)106}}.

\bibitem{mironov2010nekrasov}
A.~Mironov and A.~Morozov.
\newblock Nekrasov functions and exact {B}ohr-{S}ommerfeld integrals.
\newblock {\em Journal of High Energy Physics}, 2010(4):1--15, 2010.
\newblock \href {https://doi.org/10.1007/JHEP04(2010)040} {\path{doi:10.1007/JHEP04(2010)040}}.

\bibitem{grassi2016topological}
A.~Grassi, Y.~Hatsuda, and M.~Mari{\~n}o.
\newblock Topological strings from quantum mechanics.
\newblock {\em Annales Henri Poincar{\'e}}, 17:3177--3235, 2016.
\newblock \href {https://doi.org/10.1007/s00023-016-0479-4} {\path{doi:10.1007/s00023-016-0479-4}}.

\bibitem{dunne2017wkb}
G.~V. Dunne and M.~{\"U}nsal.
\newblock {WKB} and resurgence in the {M}athieu equation.
\newblock In {\em Resurgence, physics and numbers}, pages 249--298. Springer, 2017.
\newblock \href {https://doi.org/10.1007/978-88-7642-613-1_6} {\path{doi:10.1007/978-88-7642-613-1_6}}.

\bibitem{ito2025exact}
Katsushi Ito and Hongfei Shu.
\newblock Exact {WKB} analysis and {TBA} equations.
\newblock In {\em {ODE/IM} correspondence and Quantum Periods}, pages 23--73. Springer, 2025.
\newblock \href {https://doi.org/10.1007/978-981-96-0499-9_1} {\path{doi:10.1007/978-981-96-0499-9_1}}.

\bibitem{iwaki2014exact}
Kohei Iwaki and Tomoki Nakanishi.
\newblock Exact {WKB} analysis and cluster algebras.
\newblock {\em Journal of Physics A: Mathematical and Theoretical}, 47(47):474009, 2014.
\newblock \href {https://doi.org/10.1088/1751-8113/47/47/474009} {\path{doi:10.1088/1751-8113/47/47/474009}}.

\bibitem{del2023threefold}
F.~Del~Monte and P.~Longhi.
\newblock The threefold way to quantum periods: {WKB}, {TBA} equations and $q$-{P}ainlev{\'e}.
\newblock {\em SciPost Physics}, 15(3):112, 2023.
\newblock \href {https://doi.org/10.21468/SciPostPhys.15.3.112} {\path{doi:10.21468/SciPostPhys.15.3.112}}.

\bibitem{gaiotto2013spectral}
D.~Gaiotto, G.~W. Moore, and A.~Neitzke.
\newblock Spectral networks.
\newblock {\em Annales Henri Poincar{\'e}}, 14(7):1643--1731, 2013.
\newblock \href {https://doi.org/10.1007/s00023-013-0239-7} {\path{doi:10.1007/s00023-013-0239-7}}.

\bibitem{callot1981chasse}
J.~L. Callot, F.~Diener, M.~M. Diener, et~al.
\newblock Chasse au canard (premi{\`e}re partie).
\newblock {\em Collectanea Mathematica}, pages 37--76, 1981.

\bibitem{kristiansen2024dynamical}
K.~U. Kristiansen and P.~Szmolyan.
\newblock A dynamical systems approach to {{WKB}}-methods: {T}he simple turning point.
\newblock {\em Journal of Differential Equations}, 406:202--254, 2024.
\newblock \href {https://doi.org/10.1016/j.jde.2024.06.006} {\path{doi:10.1016/j.jde.2024.06.006}}.

\bibitem{kristiansen2025dynamical}
K.~U. Kristiansen and P.~Szmolyan.
\newblock A dynamical systems approach to {WKB}-methods: {T}he eigenvalue problem for a single well potential, 2025.
\newblock \href {https://arxiv.org/abs/2501.10707} {\path{arXiv:2501.10707}}.

\bibitem{popov2005imaginary}
V.~S. Popov.
\newblock Imaginary-time method in quantum mechanics and field theory.
\newblock {\em Physics of Atomic Nuclei}, 68:686--708, 2005.
\newblock \href {https://doi.org/10.1134/1.1903097} {\path{doi:10.1134/1.1903097}}.

\bibitem{he2012mathieu}
W.~He and Y.~Miao.
\newblock Mathieu equation and elliptic curve.
\newblock {\em Communications in Theoretical Physics}, 58(6):827, 2012.
\newblock \href {https://doi.org/10.1088/0253-6102/58/6/08} {\path{doi:10.1088/0253-6102/58/6/08}}.

\bibitem{he2015combinatorial}
W.~He.
\newblock Combinatorial approach to {M}athieu and {L}am{\'e} equations.
\newblock {\em Journal of Mathematical Physics}, 56(7), 2015.
\newblock \href {https://doi.org/10.1063/1.4926954} {\path{doi:10.1063/1.4926954}}.

\bibitem{basar2015resurgence}
G.~Basar and G.~V. Dunne.
\newblock Resurgence and the {N}ekrasov-{S}hatashvili limit: connecting weak and strong coupling in the {M}athieu and {L}am{\'e} systems.
\newblock {\em Journal of high energy physics}, 2015(2):1--49, 2015.
\newblock \href {https://doi.org/10.1007/JHEP02(2015)160} {\path{doi:10.1007/JHEP02(2015)160}}.

\bibitem{dunham1932wentzel}
J.~L. Dunham.
\newblock The {W}entzel-{B}rillouin-{K}ramers method of solving the wave equation.
\newblock {\em Physical Review}, 41(6):713, 1932.
\newblock \href {https://doi.org/10.1103/PhysRev.41.713} {\path{doi:10.1103/PhysRev.41.713}}.

\bibitem{sukhatme1999}
U.~P. Sukhatme and M.~N. Sergeenko.
\newblock Semiclassical approximation for periodic potentials, 1999.
\newblock \href {https://arxiv.org/abs/quant-ph/9911026} {\path{arXiv:quant-ph/9911026}}.

\bibitem{gorsky2018bands}
A.~Gorsky, A.~Milekhin, and N.~Sopenko.
\newblock Bands and gaps in {N}ekrasov partition function.
\newblock {\em Journal of High Energy Physics}, 2018(1), 2018.
\newblock \href {https://doi.org/10.1007/JHEP01(2018)133} {\path{doi:10.1007/JHEP01(2018)133}}.

\bibitem{seiberg1994electric}
N.~Seiberg and E.~Witten.
\newblock Electric-magnetic duality, monopole condensation, and confinement in $\mathcal{N}=2$ supersymmetric {Y}ang-{M}ills theory.
\newblock {\em Nuclear Physics B}, 426(1):19--52, 1994.
\newblock \href {https://doi.org/10.1016/0550-3213(94)90124-4} {\path{doi:10.1016/0550-3213(94)90124-4}}.

\bibitem{nekrasov2003seiberg}
N.~Nekrasov.
\newblock Seiberg-{W}itten prepotential from instanton counting.
\newblock {\em Advances in Theoretical and Mathematical Physics}, 7(5):831--864, 2003.
\newblock \href {https://doi.org/10.4310/ATMP.2003.v7.n5.a4} {\path{doi:10.4310/ATMP.2003.v7.n5.a4}}.

\bibitem{nekrasov2010quantization}
N.~Nekrasov and S.~Shatashvili.
\newblock Quantization of integrable systems and four dimensional gauge theories.
\newblock In {\em XVIth International Congress On Mathematical Physics}, pages 265--289. World Scientific, 2010.
\newblock \href {https://doi.org/10.1142/9789814304634_0015} {\path{doi:10.1142/9789814304634_0015}}.

\bibitem{krichever1994tau}
I.~M. Krichever.
\newblock The $\tau$-function of the universal {W}hitham hierarchy, matrix models and topological field theories.
\newblock {\em Communications on Pure and Applied Mathematics}, 47(4):437--475, 1994.
\newblock \href {https://doi.org/10.1002/cpa.3160470403} {\path{doi:10.1002/cpa.3160470403}}.

\bibitem{gorsky1998rg}
A.~Gorsky, A.~Marshakov, A.~Mironov, and A.~Morozov.
\newblock {RG} equations from {W}hitham hierarchy.
\newblock {\em Nuclear Physics B}, 527(3):690--716, 1998.
\newblock \href {https://doi.org/10.1016/S0550-3213(98)00315-0} {\path{doi:10.1016/S0550-3213(98)00315-0}}.

\bibitem{edelstein1999whitham}
J.~Edelstein, M.~Marino, and J.~Mas.
\newblock {W}hitham hierarchies, instanton corrections and soft supersymmetry breaking in $\mathcal{N}=2\mathrm{SU}({N})$ super {Y}ang-{M}ills theory.
\newblock {\em Nuclear Physics B}, 541(3):671--697, 1999.
\newblock \href {https://doi.org/10.1016/S0550-3213(98)00798-6} {\path{doi:10.1016/S0550-3213(98)00798-6}}.

\bibitem{codesido2019non}
S.~Codesido, M.~Marino, and R.~Schiappa.
\newblock Non-perturbative quantum mechanics from non-perturbative strings.
\newblock {\em Annales Henri Poincare}, 20:543--603, 2019.
\newblock \href {https://doi.org/10.1007/s00023-018-0751-x} {\path{doi:10.1007/s00023-018-0751-x}}.

\bibitem{kontsevich2014wall}
M.~Kontsevich and Y.~Soibelman.
\newblock {\em Wall-crossing structures in {D}onaldson-{T}homas invariants, integrable systems and mirror symmetry}, pages 197--308.
\newblock Springer, 2014.
\newblock \href {https://doi.org/10.1007/978-3-319-06514-4_6} {\path{doi:10.1007/978-3-319-06514-4_6}}.

\bibitem{gaiotto2010four}
D.~Gaiotto, G.~W. Moore, and A.~Neitzke.
\newblock Four-dimensional wall-crossing via three-dimensional field theory.
\newblock {\em Communications in Mathematical Physics}, 299(1):163--224, 2010.
\newblock \href {https://doi.org/10.1007/s00220-010-1071-2} {\path{doi:10.1007/s00220-010-1071-2}}.

\bibitem{basar2017quantum}
G.~Basar, G.~V. Dunne, and M.~{\"U}nsal.
\newblock Quantum geometry of resurgent perturbative/nonperturbative relations.
\newblock {\em Journal of High Energy Physics}, 2017(5):1--56, 2017.
\newblock \href {https://doi.org/10.1007/JHEP05(2017)087} {\path{doi:10.1007/JHEP05(2017)087}}.

\bibitem{ccavucsouglu2024resurgence}
A.~\c{C}avu\c{s}o\u{g}lu, C.~Koz\c{c}az, and K.~Tezgin.
\newblock Resurgence of deformed genus-1 curves: A novel {P/NP} relation, 2024.
\newblock \href {https://arxiv.org/abs/2408.02628} {\path{arXiv:2408.02628}}.

\bibitem{gorsky2015rg}
A.~Gorsky and A.~Milekhin.
\newblock {RG}-{W}hitham dynamics and complex {H}amiltonian systems.
\newblock {\em Nuclear Physics B}, 895:33--63, 2015.
\newblock \href {https://doi.org/10.1016/j.nuclphysb.2015.03.028} {\path{doi:10.1016/j.nuclphysb.2015.03.028}}.

\bibitem{gukov2017rg}
S.~Gukov.
\newblock {RG} flows and bifurcations.
\newblock {\em Nuclear Physics B}, 919:583--638, 2017.
\newblock \href {https://doi.org/10.1016/j.nuclphysb.2017.03.025} {\path{doi:10.1016/j.nuclphysb.2017.03.025}}.

\bibitem{Leclair2003russian}
Andr{\'e} {LeClair}, Jos{\'e}~Mar{\'i}a Rom{\'a}n, and Germ{\'a}n Sierra.
\newblock {Russian} doll renormalization group and {Kosterlitz-Thouless} flows.
\newblock {\em Nucl. Phys. B}, 675(3):584--606, 2003.
\newblock \href {https://doi.org/10.1016/j.nuclphysb.2003.09.032} {\path{doi:10.1016/j.nuclphysb.2003.09.032}}.

\bibitem{Bulycheva2014spectrum}
K.~M. Bulycheva and A.~S. Gorsky.
\newblock Limit cycles in renormalization group dynamics.
\newblock {\em Phys. Usp.}, 57(2):171--182, 2014.
\newblock \href {https://doi.org/10.3367/UFNe.0184.201402g.0182} {\path{doi:10.3367/UFNe.0184.201402g.0182}}.

\bibitem{jepsen2021rg}
C.~B. Jepsen, I.~R. Klebanov, and F.~K. Popov.
\newblock {RG} limit cycles and unconventional fixed points in perturbative {QFT}.
\newblock {\em Physical Review D}, 103(4):046015, 2021.
\newblock \href {https://doi.org/10.1103/PhysRevD.103.046015} {\path{doi:10.1103/PhysRevD.103.046015}}.

\bibitem{motamarri2024refined}
V.~Motamarri, I.~Khaymovich, and A.~Gorsky.
\newblock Refined cyclic renormalization group in {R}ussian doll model.
\newblock {\em SciPost Physics}, 17(6):157, 2024.
\newblock \href {https://doi.org/10.21468/SciPostPhys.17.6.157} {\path{doi:10.21468/SciPostPhys.17.6.157}}.

\bibitem{leclair2025non}
A.~LeClair.
\newblock Non-perturbative renormalization group for {H}iggs-like models in 4{D}, 2025.
\newblock \href {https://arxiv.org/abs/2504.09327} {\path{arXiv:2504.09327}}.

\bibitem{jepsen2021homoclinic}
C.~B. Jepsen and F.~K. Popov.
\newblock Homoclinic renormalization group flows, or when relevant operators become irrelevant.
\newblock {\em Physical Review Letters}, 127(14):141602, 2021.
\newblock \href {https://doi.org/10.1103/PhysRevLett.127.141602} {\path{doi:10.1103/PhysRevLett.127.141602}}.

\bibitem{bosschaert2022chaotic}
M.~M. Bosschaert, C.~B. Jepsen, and F.~K. Popov.
\newblock Chaotic {RG} flow in tensor models.
\newblock {\em Physical Review D}, 105(6):065021, 2022.
\newblock \href {https://doi.org/10.1103/PhysRevD.105.065021} {\path{doi:10.1103/PhysRevD.105.065021}}.

\bibitem{de2000holographic}
J.~de~Boer, E.~Verlinde, and H.~Verlinde.
\newblock On the holographic renormalization group.
\newblock {\em Journal of High Energy Physics}, 2000(08):003, 2000.
\newblock \href {https://doi.org/10.1088/1126-6708/2000/08/003} {\path{doi:10.1088/1126-6708/2000/08/003}}.

\bibitem{cherman2015decoding}
A.~Cherman, D.~Dorigoni, and M.~{\"U}nsal.
\newblock Decoding perturbation theory using resurgence: {S}tokes phenomena, new saddle points and {L}efschetz thimbles.
\newblock {\em Journal of High Energy Physics}, 2015(10):1--82, 2015.
\newblock \href {https://doi.org/10.1007/JHEP10(2015)056} {\path{doi:10.1007/JHEP10(2015)056}}.

\bibitem{basar2013resurgence}
G.~Basar, G.~V. Dunne, and M.~{\"U}nsal.
\newblock Resurgence theory, ghost-instantons, and analytic continuation of path integrals.
\newblock {\em Journal of High Energy Physics}, 2013(10):1--35, 2013.
\newblock \href {https://doi.org/10.1007/JHEP10(2013)041} {\path{doi:10.1007/JHEP10(2013)041}}.

\bibitem{aniceto2011resurgence}
I.~Aniceto, R.~Schiappa, and M.~Vonk.
\newblock The resurgence of instantons in string theory, 2013.
\newblock \href {https://arxiv.org/abs/1106.5922} {\path{arXiv:1106.5922}}.

\bibitem{provost1980riemannian}
J.~P. Provost and G.~Vallee.
\newblock Riemannian structure on manifolds of quantum states.
\newblock {\em Communications in Mathematical Physics}, 76(3):289--301, 1980.
\newblock \href {https://doi.org/10.1007/BF02193559} {\path{doi:10.1007/BF02193559}}.

\bibitem{zak1989berry}
J.~Zak.
\newblock Berry’s phase for energy bands in solids.
\newblock {\em Physical review letters}, 62(23):2747, 1989.
\newblock \href {https://doi.org/10.1103/PhysRevLett.62.2747} {\path{doi:10.1103/PhysRevLett.62.2747}}.

\bibitem{alday2010liouville}
L.~F. Alday, D.~Gaiotto, and Y.~Tachikawa.
\newblock Liouville correlation functions from four-dimensional gauge theories.
\newblock {\em Letters in Mathematical Physics}, 91(2):167--197, 2010.
\newblock \href {https://doi.org/10.1007/s11005-010-0369-5} {\path{doi:10.1007/s11005-010-0369-5}}.

\bibitem{litvinov2014classical}
A.~Litvinov, S.~Lukyanov, N.~Nekrasov, and A.~Zamolodchikov.
\newblock Classical conformal blocks and {P}ainlev{\'e} {VI}.
\newblock {\em Journal of High Energy Physics}, 2014(7):1--20, 2014.
\newblock \href {https://doi.org/10.1007/JHEP07(2014)144} {\path{doi:10.1007/JHEP07(2014)144}}.

\bibitem{lisovyy2022perturbative}
O.~Lisovyy and A.~Naidiuk.
\newblock Perturbative connection formulas for {H}eun equations.
\newblock {\em Journal of Physics A: Mathematical and Theoretical}, 55(43):434005, 2022.
\newblock \href {https://doi.org/10.1088/1751-8121/ac9ba7} {\path{doi:10.1088/1751-8121/ac9ba7}}.

\bibitem{zenkevich2011nekrasov}
Y.~Zenkevich.
\newblock Nekrasov prepotential with fundamental matter from the quantum spin chain.
\newblock {\em Physics Letters B}, 701(5):630--639, 2011.
\newblock \href {https://doi.org/10.1016/j.physletb.2011.06.030} {\path{doi:10.1016/j.physletb.2011.06.030}}.

\bibitem{aminov2022black}
G.~Aminov, A.~Grassi, and Y.~Hatsuda.
\newblock Black hole quasinormal modes and {S}eiberg-{W}itten theory.
\newblock {\em Annales Henri Poincare}, 23(6):1951--1977, 2022.
\newblock \href {https://doi.org/10.1007/s00023-021-01137-x} {\path{doi:10.1007/s00023-021-01137-x}}.

\bibitem{aminov2023black}
G.~Aminov, P.~Arnaudo, G.~Bonelli, A.~Grassi, and A.~Tanzini.
\newblock Black hole perturbation theory and multiple polylogarithms.
\newblock {\em Journal of High Energy Physics}, 2023(11):1--61, 2023.
\newblock \href {https://doi.org/10.1007/JHEP11(2023)059} {\path{doi:10.1007/JHEP11(2023)059}}.

\bibitem{bonelli2022exact}
G.~Bonelli, C.~Iossa, D.P. Lichtig, and A.~Tanzini.
\newblock Exact solution of {K}err black hole perturbations via {CFT}$_2$ and instanton counting: Greybody factor, quasinormal modes, and {L}ove numbers.
\newblock {\em Physical Review D}, 105(4):044047, 2022.
\newblock \href {https://doi.org/10.1103/PhysRevD.105.044047} {\path{doi:10.1103/PhysRevD.105.044047}}.

\bibitem{bonelli2023irregular}
G.~Bonelli, C.~Iossa, D.P. Lichtig, and A.~Tanzini.
\newblock Irregular {L}iouville correlators and connection formulae for {H}eun functions.
\newblock {\em Communications in Mathematical Physics}, 397(2):635--727, 2023.
\newblock \href {https://doi.org/10.1007/s00220-022-04497-5} {\path{doi:10.1007/s00220-022-04497-5}}.

\bibitem{ags}
A.~Alexandrov, A.~Gorsky, and L.~Senchukov.
\newblock On rotation number quantization and {QNM} of black holes.
\newblock In preparation.

\end{thebibliography}
}

\end{document}